\documentclass{amsart}
\usepackage{graphicx}
\graphicspath{ {./images/} }

\usepackage{url}
\usepackage{hyperref}
\usepackage{float}
\usepackage{multicol}

\usepackage{amsaddr, amsmath, amsthm, amssymb}

\newtheorem{proposition}{Proposition}

\def\ackname{Acknowledgements}%

\def\acknowledgement{\par\addvspace{17pt}\small\rmfamily
\trivlist\if!\ackname!\item[]\else
\item[\hskip\labelsep
{\bfseries\ackname}]\fi}

\newenvironment{acknowledgements}{\begin{acknowledgement}}
{\end{acknowledgement}}

\def\contribution{\par\addvspace{17pt}\small\rmfamily
\trivlist\if!\ackname!\item[]\else
\item[\hskip\labelsep
{\bfseries{Author contributions}}]\fi}

\newenvironment{contributions}{\begin{contribution}}
{\end{contribution}}

\title[Effectiveness of isolation measures to contain epidemics]{Effectiveness of isolation measures with app support to contain COVID-19 epidemics: a parametric approach}

\thanks{The Version of Record of this preprint is published in the \emph{Journal of Mathematical Biology}, and is available online at \url{https://doi.org/10.1007/s00285-021-01660-9}.}

\author[Andrea Maiorana, Marco Meneghelli, Mario Resnati]{Andrea Maiorana, Marco Meneghelli, Mario Resnati} % \author[2]{asdasda} % \author[3]} \date{ \today }
\address{Bending Spoons, corso Como 15, Milan, Italy }

\begin{document}

\begin{abstract}
In this study, we analyze the effectiveness of measures aimed at finding and isolating infected individuals to contain epidemics like COVID-19, as the suppression induced over the effective reproduction number. We develop a mathematical model to compute the relative suppression of the effective reproduction number of an epidemic that such measures produce. This outcome is expressed as a function of a small set of parameters that describe the main features of the epidemic and summarize the effectiveness of the isolation measures. In particular, we focus on the impact when a fraction of the population uses a mobile application for epidemic control. Finally, we apply the model to COVID-19, providing several computations as examples, and a link to a public repository to run custom calculations. These computations display in a quantitative manner the importance of recognizing infected individuals from symptoms and contact-tracing information, and isolating them as early as possible. The computations also assess the impact of each variable on the suppression of the epidemic.

\vspace{0.2cm}

\end{abstract}

\maketitle

\section{Introduction}

\subsection{Main concepts and goals}

This study aims to develop a probabilistic model to predict the effectiveness of containing an epidemic such as COVID-19 with measures aimed at finding and isolating infected individuals. More precisely, we are interested in modeling such ``isolation measures,'' by which we mean finding and isolating infected people via their symptoms and contact tracing, to predict the impact of these measures on the effective reproduction number of the epidemic. Special attention is dedicated to the case in which contact tracing is achieved, for a part of the population, through a mobile application.

Studies such as \cite{FerrEtAl2020Quant} have underlined the role of asymptomatic and presymptomatic transmission in the COVID-19 outbreak, and the consequent importance of using a mobile application for efficient contact tracing. This insight has also led to the development of models to quantitatively assess the impact of a contact tracing app on the epidemic, primarily through agent-based approaches like in \cite{OxfordRepo}.

In this paper, we propose an analytical approach to answer the following questions: How is the effective number $R_t$ of an epidemic impacted when isolation measures are in place versus when they are not, and what are the main factors contributing to the suppression of $R_t$? We take the effective reproduction number \emph{in the absence} of isolation measures, denoted by $R^0_t$, as an input of our model, which is thus independent of any underlying epidemic model. Moreover, our approach is \emph{parametric} in that we concentrate the quantitative description of the isolation measures into relatively few, comprehensible parameters that comprise the input of the model. These parameters include the share of the population using an app, the share of people who self-isolate upon testing positive, and more.

Previous studies concerning the impact on the epidemic of isolating infected individuals include \cite{muller2000contact}, which proposes a generative stochastic model of SIR-type, and \cite{fraser2004factors}, which uses an analytical method more similar to our own.
\bigbreak

The starting point of our analysis is the \emph{effective reproduction number} $R^0_t$ \emph{in the absence of isolation measures},\footnote{$R^0_t$ must not be confused with the \emph{basic reproduction number} $R_0$.} that we consider as given. When speaking about modeling ``isolation measures,'' we refer to policies focused on finding and selectively isolating infected individuals. We do not refer to generalized actions like imposing a lockdown, whose impact on the epidemic is considered already known and encompassed in $R^0_t$.

$R^0_t$ is defined, for any absolute time $t$, as the expected number of cases generated by a random individual who was infected at time $t$ during their lifetime. This quantity can be written as an integral
$$R^0_t=\int_{[0,+\infty)}\beta^0_t(\tau)\mathrm{d}\tau\,,$$
where $\beta^0_t$ is the \emph{infectiousness} (also called \emph{effective contact rate}): $\beta^0_t$ is a function describing the expected number of cases generated by an individual infected at time $t$, per unit of \emph{infectious age}, that is the period of time (measured in days) elapsed from the time of infection of the individual. So, for example, the number
$$\int_{[1,3)}\beta^0_t(\tau)\mathrm{d}\tau$$
is the expected number of people infected between 24 and 72 hours from the infector's moment of infection.
Note that the normalization $\beta^0_t/R^0_t$ is the PDF of the \emph{generation time}, the time taken by an individual infected at $t$ to infect a different individual.\footnote{This is better explained in \S\ref{A: generation time}.}
\bigbreak
 
In this study, we set up a methodology and a model to analyze changes in the reproduction number when the population is subject to isolation measures, including the support of an app for individuals who have tested positive, and depending on some parameters of simple interpretation. We denote by
$$R_t=\int_{[0,+\infty)}\beta_t(\tau)\mathrm{d}\tau$$
the effective reproduction number \emph{in presence of isolation measures}, and we compute $R_t$ as a function of $R^0_t$,\footnote{We stress that the time evolution of $R^0_t$, describing how the epidemics would have evolved without the measures we are modeling, is taken as known\textendash our goal is to study the relative impact of the measures. In particular, we do not take into account possible second-order effects on $R^0_t$, such as general changes in the behavior of the population, that may come as a consequence of the measures and their impact.} other epidemiological data such as the symptom onset distribution, and some parameters describing the isolation measures, such as the probability that an infected, symptomatic individual gets a test, or the probability that a recipient of the infection gets notified when their infector receives a positive test. We only model how isolation measures work and how they affect the epidemic, without assuming anything about how the epidemic itself develops. In particular, our model is agnostic of any particular form for $\beta^0_t$ and $R^0_t$.

The final goal of the model we propose is to understand the most important leverages that may facilitate optimization to better direct efforts of decision-makers, scientists, and developers. Such factors include app efficiencies, timeliness of notifications, app adoption in the population, and others.

\subsection{The assumptions of the model and outline of the paper}\label{The assumptions of the model}

The model developed in \S\ref{Mathematical model} is the translation into mathematical terms of the following assumptions, that describe an idealized schema in which infected individuals acknowledge their illness and take measures to avoid infecting others.

\begin{itemize}
\item An infected individual who shows symptoms is immediately\footnote{This and the following assumption of immediate notification simplify the treatment to the extent that they avoid adding further distributions modeling some delays. In fact, they are not essential hypotheses, and such real-world delays could be also taken into account in the current setting, by including them in the distribution of the time $\Delta^{\text{A}\to\text{T}}$ between notification and test, introduced below.} notified that they should take a test (which does not discount the possibility that they acknowledge this necessity independent of an external input). This process does not always necessarily occur, but does so with a probability $s^\text{s}$.
\item Given a infector-infectee pair, when the infector tests positive after the contagion, the infectee is immediately notified to take a test, with probability $s^\text{c}$.
\item In either scenario, after an infected individual is notified to take a test, they take a test which will return a positive result after a time from the notification, which is distributed according to a given distribution $\Delta^{\text{A}\to\text{T}}$ (possibly reaching $+\infty$ to account for the case in which the individual is never tested or never receives the positive outcome).
\item Immediately upon receiving the positive outcome of the test, an average infected individual will self-isolate with probability $\xi$. Put differently, the number of individuals they infect from this moment is reduced by a factor $1-\xi$ compared to the scenario in which they do not take any isolation measures.
\end{itemize}

The equations derived from these hypotheses produce an algorithm that computes the time evolution of the key quantities. This is summarized in \S\ref{Summary of eqns homog model}.

Note that in our model we are only considering \emph{forward} contact tracing, i.e.,\ infectees are notified of the positive result of their infectors, but not vice-versa. Doing otherwise would significantly complicate the discussion. This is probably the main limitation of the model, which may thus underestimate the effectiveness of the isolation measures: While backward contact tracing is in general less effective when timeliness in isolating infected individuals is key, it must be noted that its effect may be significant for epidemics for which super-spreaders, i.e.\ individuals that infect a large number of people, have a major impact on the contagion. Such individuals may be identified more easily thanks to backward tracing. A treatment of backward tracing in the context of a generative model is covered in \cite{muller2000contact}, \S3.1.
 
Subsequently, in \S\ref{Modifications with app} we consider a more complex model. Instead, we assume that the population is split into two groups, depending on whether or not they use a mobile application for epidemic control. The parameters $s^\text{s}$ and $s^\text{c}$ are different, depending on whether they refer to individuals who use the app.

Finally, in \S\ref{Calculations}, we apply these models by computing the suppression of $R_t$ for specific choices of the input parameters, particularly to assess the importance of such parameters. As for the input parameters that describe the epidemic, we use data relative to COVID-19. All these data are taken for a single source, \cite{FerrEtAl2020Quant}. It should be noted that these quantities are still preliminary, have quite large uncertainties, and are not necessarily the most up-to-date. However, we stress that these data are only used as inputs in all our computations, which can be easily reproduced and extended by using the code available in the open repository \cite{mm2021epidemics}. It would be immediate to redo the computations with different inputs, to reflect any new understandings the scientific community should gain on COVID-19. In addition to this, in \S\ref{Dependency on epidemic data} we briefly check the robustness of our results with respect to changes in some epidemic data, namely the share of infected individuals that are asymptomatic, the contribution of those individuals to the reproduction number, and the generation time distribution.

The paper includes an Appendix where the main steps of the mathematical model are proven rigorously, in a framework where the hypotheses can be formulated precisely using the language of probability theory.

\subsection{Discussion of the results}

Summing up, this paper introduces a model of targeted isolation measures\textendash with special attention paid to those based on contact tracing\textendash in the context of an epidemic with given dynamics. It studies the impact of these, measured as the change in the key indicators of the epidemics (first of all, the reproduction number) with respect to the situation without measures. It presents a methodology to turn the assumptions defining the model into mathematical equations, without assuming an underlying model explaining the time evolution of the epidemic. In particular, the formalism developed in the Appendix allows a careful and exact development of the theory, in which all the interdependencies of the involved quantities are clarified. We end up with with a set of equations that express the relevant quantities in terms of those relative to previous times, giving a deterministic time evolution.

These equations (summarized in \S\ref{Summary of eqns homog model} for the ``homogeneous'' setting) are quite complex, reflecting the non-triviality of the assumptions about how isolation measures work. This makes it hard to analyze them analytically, for example, to study the asymptotic behaviour of the solutions, as was done in \cite{fraser2004factors}. On the other hand, our treatment allows us to refrain from making strong and unrealistic independence assumptions about the involved quantities, and leaves us greater freedom in setting up the hypotheses of how contact tracing works (for example, the isolation of contact-traced individuals is not assumed to be certain, nor immediate). And, notably, it allows us to numerically compute, with arbitrary precision, the time evolution of the reproduction number $R_t$ (and, hence, of the epidemic size) starting from the ``default'' reproduction number $R^0_t$, other epidemiological data, and the parameters introduced in \S\ref{The assumptions of the model} describing the isolation measures.

We stress that, despite our extensive use of the language of probability theory, our model of the isolation measures is deterministic: It works as if the full history of the epidemic, with or without isolation measures, is given, and uses some parameters describing the mean efficacy of the isolation measures on the population. It then expresses $R_t$ in terms of $R^0_t$ and these parameters.

A limitation to the model comes from such homogeneous-mixing hypotheses regarding contact tracing and isolation policies: The only heterogeneity taken into account is the separation between individuals who do or do not use an app in \S\ref{Modifications with app}. For example, the fact that, in reality, individuals belonging to the same household are more easily traced (in addition to being more easily infected by each other) is not taken into account. Besides the absence of backward contact tracing, mentioned in \S\ref{The assumptions of the model}, other limitations may be attributed to the specific form of the hypotheses. However, many changes to the assumptions could be taken into account within the same mathematical framework: Features such as a different delay in testing for symptomatic or contact-traced individuals, or the existence of a targeted quarantine for potential infected individuals (even before they get tested) could be modeled without adding conceptual complications.

By using the model in \S\ref{Calculations} to compute the suppression of $R_t$, we can recognize how isolation measures, particularly app-mediated isolation measures, can play an important role in suppressing epidemics like COVID-19. However, our results show how the impact of such measures is strongly sensitive to parameters describing their efficiency and timeliness: For example, the reduction in $R_t$ quickly becomes insignificant as the time taken to get a positive test result (and then to start isolating) grows past a few days (see Fig.\ \ref{fig:timeliness_span}).

The computations relative to the case in which an app is used show the importance of having an app which is effective at spotting infections, maximizing the fraction of true-positives.\footnote{To be trusted by its users, the app should also aim at reducing the fraction of false-positives. This is something that our study does not consider.} Past studies like \cite{bendavid2020covid} and \cite{li2020early} suggest that ``standard'' contact tracing measures used by healthcare systems may be less efficient (fewer truly infected individuals are recognized) and slower when compared to an app (usually, several days elapse between symptom onset, the first medical visit, and the test outcome). In the computations, we model this fact by setting different parameters for people using an app and people who don't, with the latter parameters left to reasonably low values. We analyze how the impact on the epidemics depends on these parameters and the app adoption rate (Fig.\ \ref{fig:adoption_span}), showing how these are all key factors in reaching satisfactory epidemic suppression levels.

\section{The mathematical model in the homogeneous population setting}\label{Mathematical model}

In this section, we develop the core mathematical model of the paper. We do so with a simplified scenario in which the same isolation measures apply to the entire population, thus eliminating the need to distinguish between those who do and who do not use an app. Some mathematical derivations require extra care, and their complete proofs have been moved to the Appendix to prevent this section from being loaded with many formulae and a heavier formalism.

\subsection{Notations and conventions}

We consider random variables on the sample space of all infected individuals, describing (absolute) times at which certain events happen: $t^\text{I}$ (time of infection), $t^\text{S}$ (time of symptom onset), $t^\text{A}$ (time of infection notification), $t^\text{T}$ (time of positive test). These variables can take $+\infty$ as a value to express the cases in which an event never takes place (this is useful when writing relations between them).

As we want to relate these variables to the reproduction number $R_t$, which measures the average number of people infected by an individual infected at a given time $t$, it is logical that all these variables refer to the infectious age (that is, the time from the infection) of the average individual infected at $t$: so we have, for example, the relative time of symptom onset, which is the $[0,+\infty]$-valued random variable
$$\tau^\text{S}_t
=(t^\text{S}-t^\text{I})|_{t^\text{I}=t}
=t^\text{S}|_{t^\text{I}=t}-t\,.$$
We can assume that this variable is independent of the contagion time $t$. Hence, we denote it by $\tau^\text{S}$. Analogously, we have the random variables $\tau^\text{A}_t$ (time of notification for an individual infected at $t$, measured since $t$), $\tau^T_t$ (time of positive test for an individual infected at $t$, measured since $t$). 

In this section we need to understand how to describe the random variables $\tau^\text{S}$, $\tau^\text{A}_t$, $\tau^\text{T}_t$, and their relation to the reproduction number
$$R_t=\int_{[0,+\infty)}\beta_t(\tau)\,\mathrm{d}\tau\,,$$
based on the assumptions of \S\ref{The assumptions of the model}. The finite parts of these random variables are described using improper CDFs, denoted by $F^\text{S}$, $F^\text{A}_t$, and $F^\text{T}_t$ respectively, whose limit for $\tau\to+\infty$ (representing the probability that each time is less than infinite) may be less than 1. So, for example, $F^\text{T}_t(\tau)$ denotes the probability that an individual infected at $t$ tests positive within a time $\tau$ from the time of infection. $\lim_{\tau\to+\infty}F^\text{T}_t(\tau)$ is the probability that the same individual eventually tests positive.

Further auxiliary variables are introduced later on.

\subsection{The suppression model for \texorpdfstring{$R_t$}{Rt}}\label{Rt suppression model}

Recall from \S\ref{The assumptions of the model} how we assume that self-isolation works: If an infected individual tests positive, then they immediately self-isolate, resulting in a reduction, on average, of the number of people they subsequently infect by a multiplicative factor $1-\xi_t$, which we assume given, and possibly depending on the time $t$ at which the individual was infected.\footnote{Equivalently, this hypothesis could be viewed as the assumption that an individual who tests positive either self-isolates completely, without infecting anyone else from that moment, or, alternatively, does nothing, with the first circumstance happening with probability $\xi_t$.}

We can then determine a relation between the ``default" reproduction number density $\beta^0_t$, its correction $\beta_t$ as a result of the isolation measures, and the distribution of the relative time $\tau^\text{T}_t$ at which individuals infected at $t$ receive a positive test result. This relation holds for any $t$ greater or equal to the time $t_0$ at which the isolation measures are enacted.

For simplicity, let's assume for a moment that receiving a positive test and infecting someone (assuming no isolation measures) at a given infectious age $\tau$ are independent events. By $\tau$, an individual who was infected at $t$ has already received a test with probability $P(\tau^\text{T}_t<\tau) = F^\text{T}_t(\tau)$. In such a case, the number of people they infect per unit time is $(1-\xi_t)\beta^0_t(\tau)$. Alternatively, if the individual has not received a test by $\tau$ (which happens with probability $1-F^\text{T}_t(\tau)$), they do not self-isolate, and the average number of people they infect per unit time is just $\beta^0_t(\tau)$. In summary, we have, for any $\tau\in[0,+\infty)$,
\begin{equation}\label{Eq: beta_t corr indip ass}
\begin{array}{rcl}
\beta_t(\tau) &=& F^\text{T}_t(\tau)(1-\xi_t)\beta^0_t(\tau) + \left(1-F^\text{T}_t(\tau)\right)\beta^0_t(\tau)\\
&=& \beta^0_t(\tau)\left(1-\xi_t\,F^\text{T}_t(\tau) \right)\,.
\end{array}
\end{equation}
This is analogous to Eq.\ 6 in \cite{fraser2004factors}. To illustrate further, suppose that \emph{all} infected individuals test positive at the same infectious age $\tau_\text{T}$, i.e.\ $F^\text{T}_t(\tau)$ is a Heaviside function with step at $\tau_\text{T}$: then we have $\beta_t(\tau)=\beta^0_t(\tau)$ for $\tau<\tau_\text{T}$ and $\beta_t(\tau)=(1-\xi_t)\,\beta^0_t(\tau)$ for $\tau\geq\tau_\text{T}$.

However, the above result relies on the assumption of independence between testing positive and the number of people the individual would infect without isolation. In practice, this is not an adequate reflection of what occurs. For example, with COVID-19, it is known that a significant proportion of the infected population is asymptomatic, and less contagious\textendash see e.g.\ \cite{mizumoto2020estimating}, \cite{FerrEtAl2020Quant}. Given the lack of symptoms, this population has a lower probability of self-isolating. To overcome this factor, we introduce a new random variable $G$, which has a finite range $\{g_1,...,g_n\}$ that describes the severity of symptoms of an infected individual. It is assumed to be independent of the time $\tau^\text{S}$ of symptom onset, but it is related to the number of infected people and the probability of the individual recognizing their own symptoms. Then, to write a relation between $F^\text{T}_t$ and $R_t$, we restrict the relevant random variables to each possible value of $G$: for any $g=g_1,...,g_N$ we denote by
$$F^\text{T}_{t,g}(\tau)$$
the probability that an individual infected at time $t$ and with severity $g$ has tested positive by $\tau$. Similarly, we denote by
$$R_{t,g}=\int_{[0,+\infty)}\beta_{t,g}(\tau)\,\mathrm{d}\tau$$
the average number of people infected by an individual infected at $t$ and with severity $g$, and by $R^0_{t,g}$ the analogous quantity in absence of isolation measures. Assuming now that \emph{for a given $g$} the number of people infected (without isolation) and the event of being tested are independent, we write our ``suppression formula'' as
\begin{equation}\label{Eq: beta_t,g corr}
\begin{array}{c}
\beta_{t,g}(\tau) = \beta^0_{t,g}(\tau)\left(1-\xi_t\,F^\text{T}_{t,g}(\tau) \right)\,.
\end{array}
\end{equation}
In \S\ref{A: suppression formula} we include a careful derivation of this formula. Note that the relations with the aggregate variables are
$$
F^\text{T}_t=\sum_g p_g F^\text{T}_{t,g}\,,\quad 
R_t=\sum_g p_g R_{t,g}\,,
$$
where $p_g=P(G=g)$ is the probability that an infected individual has symptoms with severity $g$.\footnote{Note that, according to our convention, $R_t$ is the \emph{weighted average} of its components $R_{t,g}$. Often, in the literature (e.g.\ in \cite{FerrEtAl2020Quant}) a different convention is used, according to which the components \emph{sum} to $R_t$. To switch to the latter convention, each $R_{t,g}$ should be divided by the respective probability $p_g$.}

Also, in \S\ref{Calculations} we always take $G$ to assume the values 0 and 1 only, to describe asymptomatic versus symptomatic infected individuals. However, this formalism allows for a greater diversification of $R^0_t$, according to the severity of the illness.

We end this subsection with an example of an application of \eqref{Eq: beta_t,g corr} in a simplified scenario. Suppose that $G$ only takes the values 0 and 1, describing asymptomatic and symptomatic infected individuals, and that each constitutes half of the population. Suppose also that $\xi_t=1$, and that asymptomatic individuals are never tested, so that $F^\textup{T}_{t,0}=0$, while symptomatic individuals are tested immediately after infection, so that $F_{t,1}^\textup{T}(\tau)=\theta(\tau)$, where $\theta$ is the Heaviside function. Then, we have $R_{t,0}=R^0_{t,0}$ and $R_{t,1}=0$, so that $R_t=R^0_{t,0}/2$. Had we used Eq.\ \ref{Eq: beta_t corr indip ass} instead, we would have ended up with $R_t=R^0_t/2$, which does not take into account the fact that isolating symptomatic individuals has a greater impact on the reduction of $R_t$ than isolating the same proportion of randomly chosen individuals.
  
\subsection{First considerations on the variables \texorpdfstring{$\tau^\textup{S}$, $\tau^\textup{A}_t$, and $\tau^\textup{T}_t$}{tauSt, tauAt, and tauTt}}\label{tauSt, tauAt, and tauTt}

The distribution of the time $\tau^S$ of symptom onset is independent of the isolation policy and is considered as given throughout the paper, although its specific shape is irrelevant in this section.\footnote{In \S\ref{Calculations} we take $F^\text{S}$ to be a log-normal distribution, following the literature.}

The description of $\tau^\text{A}_t$ is addressed in the next subsection. Here, we only consider its relation with $\tau^\text{T}_t$: Having assumed that the time between notification and testing positive is described by a given random variable $\Delta^{\text{A}\to\text{T}}$, which is independent from $\tau^\text{A}_t$ and for simplicity constant in absolute time, we have
$$\tau^\text{T}_t = \tau^\text{A}_t + \Delta^{\text{A}\to\text{T}}\,.$$
The relation still holds if we restrict it to individuals with a given severity $g$, and hence
\begin{equation}\label{Eq: rel FA, FT}
F^\text{T}_{t,g}(\tau) = \int_{[0,+\infty)}F^\text{A}_{t,g}(\tau-\tau')\,\mathrm{d}F^{\text{A}\to\text{T}}(\tau')\,, 
\end{equation}
where $F^{\text{A}\to\text{T}}$ is the improper CDF of $\Delta^{\text{A}\to\text{T}}$.

\subsection{Describing \texorpdfstring{$\tau^\textup{A}_t$}{tauAt}}\label{Describing tauAt}

In this subsection, we consider the random variable $\tau^\textup{A}_t$ and study the relations with it that formalize the assumptions of \S\ref{The assumptions of the model}, namely:
\begin{itemize}
    \item When an infected individual shows symptoms, they receive an immediate notification to get tested, with probability $s^\text{s}_{t,g}$ depending on the severity $g$ of symptoms, and possibly on the infection time $t$.
    \item Immediately after an infector tests positive, each infectee is notified of the risk, with probability $s^\text{c}_t$. If the contagion takes place after the positive test, then the infectee is never notified.
\end{itemize}

We introduce two new random variables relative to individuals infected at a given time $t$, describing the receiving of a notification for either cause:
\begin{itemize}
    \item We denote by $\tau^\text{A,s}_{t,g}$ the time from infection at which an individual infected at $t$ and with severity $g$ is notified because of symptoms. We assume that this happens with probability $s^\text{s}_{t,g}$ at the time $\tau^S$ of the symptom onset, so its improper CDF is simply\footnote{For simplicity, we use a unique distribution $F^\text{S}$ for all degrees of severity. For asymptomatic individuals, $s^\text{s}_{t,g}$ would be equal to 0.}
    \begin{equation}\label{Eq: CDF FAs}
    F^\text{A,s}_{t,g} = s^\text{s}_{t,g}\, F^\text{S}\,.
    \end{equation}
    \item We denote by $\tau^\text{A,c}_t$ the time from infection at which an individual infected at $t$ receives a notification resulting from the positive test of their infector. Below, we see how to describe this.
\end{itemize}
The relation between these new variables and $\tau^\text{A}_{t,g}$ is
$$\tau^\text{A}_{t,g} =\min(\tau^\text{A,s}_{t,g}, \tau_t^\text{A,c})\,.$$
In terms of improper CDFs, and assuming independence of the two notification times, this gives
\begin{equation}\label{Eq: rel CDFs FA}
F^\text{A}_{t,g} = F^\text{A,s}_{t,g} + F_t^\text{A,c} - F^\text{A,s}_{t,g}F_t^\text{A,c}\,.
\end{equation}

Describing $\tau^\text{A,c}_t$ requires the introduction of an additional random variable $\tau^\sigma_t$, that gives, for any individual infected at $t$, the time elapsed between the the infection time of their infector and $t$. In particular, we need the joint distribution of $\tau^\sigma_t$ and the severity $G$, that can be described in terms of improper CDFs $F^{\sigma, g}_t$: Let
$$
F^{\sigma, g}_t(\tau)
$$
denote the probability that, given an individual infected at $t$, their infector has severity $g$ and was infected at a time $t'\geq t - \tau$. Note that these improper CDFs satisfy a normalization condition
$$
\lim_{\tau\to+\infty}\sum_g F^{\sigma, g}_t(\tau) = 1\,,
$$
and they are completely determined by quantities relative to times preceding $t$, namely the number of infected people and the infectiousness (more details on how they are computed are deferred to \S\ref{A: generation time}).

Now, the notification time $\tau^\text{A,c}_t$ of an individual infected at $t$ is by hypothesis equal to the testing time $\tau^\text{T}_{t'}$ of the infector minus the generation time $\tau^\sigma_t$, but only if the notification actually occurs, which happens with probability $s_t^\text{c}$ provided that the contagion took place before $\tau^\text{T}_{t'}$. Hence, to get the improper CDF $F^\text{A,c}_t$ we should first average $F^\text{T}_{t-\tau}$, translated to the left by $\tau=\tau^\sigma_t$, over all possible values of $\tau>0$, each weighted by the probability of the generation time being $\tau$. In doing this we should also treat separately the different severity levels that the infector may have, as these impact the testing time distribution. So $F^\text{A,c}_t(\rho)$ should look like a sum
$$
s_t^\text{c}\sum_g\int_{(0, +\infty)}F^\text{T}_{t-\tau,g}(\rho+\tau)\,\mathrm{d}F^{\sigma, g}_t(\tau)\,.
$$
This formula doesn't take into account that by assumption the notification can only occur after the contagion time, meaning that $F^\text{A,c}_t$ must be supported on positive numbers. This is considered by replacing the integrand with the probability
$$
P(\tau<\tau^\text{T}_{t-\tau,g}\leq\rho+\tau) = F^\text{T}_{t-\tau,g}(\rho+\tau) - F^\text{T}_{t-\tau,g}(\tau)\,.
$$
Also, in averaging the CDFs $F^\text{T}_{t-\tau}$ we should take into account the fact that the testing time of the infector is not distributed like the testing time of an arbitrary individual: Having infected someone at the infectious age $\tau$, the infector is more likely than average to be tested \emph{after} $\tau$, or to never receive a test. As we will show carefully in the Appendix, to take this into account we need to divide the integrand by the same suppression factor $1-\xi_t\,F^\text{T}_{t,g}(\tau)$ that appears in Eq.\ \eqref{Eq: beta_t,g corr}, evaluated at $t-\tau$. We conclude that, for any $\rho\geq 0$, we have
\begin{equation}\label{Eq: time evolution equation}
    F_t^\text{A,c}(\rho) = 
s^\text{c}_t\,\sum_g\int_{(0, +\infty)}
    \frac{F^\text{T}_{t-\tau,g}(\rho + \tau) - F^\text{T}_{t-\tau,g}(\tau)}
    {1-\xi_{t-\tau} F_{t-\tau,g}^\text{T}(\tau)}\,
    \mathrm{d}F^{\sigma,g}_t(\tau)\,.
\end{equation}
This result is proven rigorously in \S\ref{A: time evolution}.

\subsection{Summary and discrete-time algorithm}\label{Summary of eqns homog model}

In this section, we have translated the hypotheses made in \S\ref{The assumptions of the model} into mathematical equations describing a dynamical system. In doing this, we added a few natural assumptions of independence between the variables under considerations, namely:
\begin{itemize}
    \item the assumption in \S\ref{Rt suppression model} that the testing time of an individual \emph{with given severity} is independent from their \emph{default} infectiousness
    \item the assumption of independence between notification times $\tau^\text{A,s}_{t,g}$ and $\tau^\text{A,c}_t$ and the testing delay $\Delta^{\text{A}\to\text{T}}$
\end{itemize}
Putting all the equations together, we see that we can compute, at any time $t$, the suppressed infectiousness $\beta_{t,g}$ in terms of the parameters $s^\text{s}_{t,g}$, $s^\text{c}_t$, $\xi_t$, $\Delta^{\text{A}\to\text{T}}$ of the model, the default infectiousness $\beta^0_{t,g}$ and the other known epidemiological quantities, and the distributions relative to previous times $t'<t$.

To add an initial condition to the dynamical system, we assume that the isolation measures start at a given absolute time $t_0$, so that $s^\text{s}_t = s^\text{c}_t = \xi_t = 0$ for $t<t_0$.\footnote{Note that this doesn't prevent us from modeling isolation measures gradually put into place, which can be done simply by taking these parameters to be continuous in $t$.} Hence, all individuals infected at $t<t_0$ will never take a test (even after $t_0$) and never self-isolate. As a consequence, the effective reproduction number is $R^0_t$ for $t<t_0$, while it gets reduced according to Eq.\ \eqref{Eq: beta_t,g corr} for $t\geq t_0$. In particular, individuals infected at $t=t_0$ can only be notified of the need to take a test through symptoms, so that $F^\text{A,c}_{t_0}=0$. 

Our set of equations can be approximated with arbitrary precision to a discrete-time algorithm that computes how the epidemic evolves, given the above data.\footnote{In this discrete setting, all the integrals appearing in the equations reduce to finite sums. In fact, our approach in the Appendix is to derive the same equations starting from a discrete probabilistic model.} This is the algorithm used in the calculations of \S\ref{Calculations}.

Summing up, for each time $t\geq t_0$, the algorithm works as follows:

\begin{enumerate}
    \item Compute the number $\nu_t$ of individuals infected at $t$ and the improper CDFs $F^{\sigma,g}_t$, from $\nu_{t'}$ and $\beta_{t',g}$ for $t' < t$ (as detailed in \S\ref{A: generation time}).
    \item Compute the distribution of $\tau^\text{A,s}_{t,g}$ using Eq.\ \eqref{Eq: CDF FAs}.
    \item Compute the distribution of $\tau^\text{A,c}_t$ from $F^{\sigma,g}_t$ and the distribution of $\tau^\text{T}_{t',g}$, for $t'<t$, using Eq.\ \eqref{Eq: time evolution equation}. If $t=t_0$, just take $F^\text{A,c}_t=0$.
    \item Compute the distribution of $\tau^\text{A}_{t,g}$ using Eq.\ \eqref{Eq: rel CDFs FA}, and then the distribution of $\tau^\text{T}_{t,g}$ via Eq.\ \eqref{Eq: rel FA, FT}.
    \item Compute $\beta_{t,g}$ using the distribution of $\tau^\text{T}_{t,g}$, via Eq.\ \eqref{Eq: beta_t,g corr}.
\end{enumerate}

\section{The extended model including the use of an app for epidemic suppression }\label{Modifications with app}

So far, we have operated under the hypothesis that the ability to inform infected people that their source has been infected can be described by a single (possibly time-dependent) parameter $s_t^\text{c}$. Now, let's suppose that the population is divided into people who use an app for epidemic control and people who do not. This forces us to complicate the model of \S\ref{Mathematical model} because, when we analyze the distribution of the notification time $\tau^\text{A,c}_t$ for people with the app, we need to apply different weights to the cases in which the source of the contagion has the app or does not. We also leave open the possibility that people using the app may have a different probability $s_t^\text{s}$ of requiring a test because of their symptoms.

The generalization of the homogeneous scenario to this case is quite straightforward. In any case, some more mathematical detail has been added in \S\ref{A: scenario with app}.

\subsection{Parameters and random variables in the two-component model}

A share $\epsilon_{t,\text{app}}$ of the infected population, perhaps depending on the absolute time $t$, uses an app that may do the following:
\begin{itemize}
    \item It gives the users clear instructions on how to behave when they have symptoms indicative of the disease, assuming that this can increase the probability that an infected individual asks the health authorities to be tested because of their symptoms.
    \item It notifies the users when they have had contact with an infected individual who also uses the app, assuming that this can increase the probability that an infected individual asks the health authorities to be tested because of contact with an infected person.
\end{itemize}
 
We then distinguish $s^\text{s}_{t,g}$ into $s^\text{s,app}_{t,g}$ and $s^\text{s,no app}_{t,g}$, describing the probability that an individual infected at $t$, respectively with or without the app, is notified of the need to be tested given that they have symptoms with severity $g$. Note that 
\begin{equation}\label{Eq: sS app, no app}
s^\text{s}_{t,g}=\epsilon_{t,\text{app}}\, s^\text{s,app}_{t,g} + (1-\epsilon_{t,\text{app}})\, s^\text{s,no app}_{t,g}\,,
\end{equation}
so that this distinction does not complicate the model, and is made only for adding clarity in the computations.

The increased complexity of this situation lies in the fact that $s^\text{c}_t$ now has to be replaced by two parameters $s^\text{c,app}_t$ and $s^\text{c,no app}_t$, describing the probabilities that, given an infector-infectee pair, the positive testing of the infector occurred after the infection caused a notification to be sent to the infectee, respectively in the cases that \emph{both} the infector and the infectee have the app, and that \emph{at least one} of them does not have the app. Note that there is no relation between $s^\text{c,app}_t$ and $s^\text{c,no app}_t$ and the general $s^\text{c}_t$ as simple as Eq.\ \eqref{Eq: sS app, no app}.

We also distinguish each random variable between people with the app and people without it. For example, the time of notification due to contact now reads $\tau_{t,\text{app}}^\text{A,c}$ for people with the app and $\tau_{t,\text{no app}}^\text{A,c}$ for people without it. The relation between their improper CDFs is
$$F_t^\text{A,c} = \epsilon_{t,\text{app}}\,F_{t,\text{app}}^\text{A,c} + (1-\epsilon_{t,\text{app}})\,F_{t,\text{no app}}^\text{A,c}\,.$$
We have analogous formulae for $\tau^T_t$ and $\tau^\text{A,s}_t$, while there is no need to make a distinction for $\tau^S$.

Likewise, we have to separate $R_t$ into two components $R_{t,\text{app}}$ and $R_{t,\text{no app}}$, namely, the average number of people infected by someone infected at $t$ who has or does not have the app, respectively:
$$R_t = \epsilon_{t,\text{app}}\, R_{t,\text{app}} + (1-\epsilon_{t,\text{app}})\, R_{t,\text{no app}}\,.$$
Analogous relations hold when restricted to individuals whose illness has a given severity $g$.

It is reasonable to assume that having or not having the app is independent of symptom severity, so that, for example, the fraction of individuals infected at time $t$ using the app and with severity $g$ is $\epsilon_{t,\text{app}}p_g$. Also, while of course having an app does impact the testing time distribution and the infectiousness, we can safely suppose that it is independent of the \emph{default} infectiousness, i.e.\ the number of people an individual would have infected in the absence of measures. This is why in this scenario the suppression formula \eqref{Eq: beta_t,g corr} simply becomes
\begin{equation}\label{Eq: beta_t,g_app corr}
\begin{array}{c}
\beta_{t,g,a}(\tau) = \beta^0_{t,g}(\tau)\left(1-\xi_t\,F^\text{T}_{t,g,a}(\tau) \right)\,,
\end{array}
\end{equation}
for $a=\text{app},\text{no app}$.

\subsection{The mathematical relations between the random variables}

Now, we can write the new relations between the random variables. Eq.\ \eqref{Eq: CDF FAs} is replaced by
\begin{equation*}
F^\text{A,s}_{t,g,\text{app}} = s^\text{s,app}_{t,g}F^\text{S}\,,\quad
F^\text{A,s}_{t,g,\text{no app}} = s^\text{s,no app}_{t,g}F^\text{S}\,.
\end{equation*}
The relations \eqref{Eq: rel FA, FT}, \eqref{Eq: rel CDFs FA} immediately extend to each component.

The distributions of $\tau^\text{A,c}_{t,\text{app}}$ and $\tau^\text{A,c}_{t,\text{no app}}$ can be computed similarly to as we did in \S\ref{Describing tauAt} for the homogeneous case. But now, for each of them Eq.\ \eqref{Eq: time evolution equation} needs to be split into two parts, accounting for the cases in which the source of the infection has or doesn't have the app:
\begin{equation}\label{Eq: time evolution equation with app}
\begin{array}{rcl}
    F_{t,\text{app}}^\text{A,c}(\rho) 
    &=& 
s^\text{c,app}_t\,
\sum_g\int_{(0,+\infty)}
    \frac{F^\text{T}_{t-\tau,g,\text{app}}(\rho + \tau) - F^\text{T}_{t-\tau,g,\text{app}}(\tau)}
    {1-\xi_{t-\tau} F_{t-\tau,g,\text{app}}^\text{T}(\tau)}\,
    \mathrm{d}F^{\sigma,g,\text{app}}_t(\tau)\\
    &&
    + s^\text{c,no app}_t\,
\sum_g\int_{(0,+\infty)}
    \frac{F^\text{T}_{t-\tau,g,\text{no app}}(\rho + \tau) - F^\text{T}_{t-\tau,g,\text{no app}}(\tau)}
    {1-\xi_{t-\tau} F_{t-\tau,g,\text{no app}}^\text{T}(\tau)}\,
    \mathrm{d}F^{\sigma,g,\text{no app}}_t(\tau)
    \,.\\
\end{array}    
\end{equation}
For $F_{t,\text{no app}}^\text{A,c}$ we get a similar equation with $s^\text{c,app}_t$ replaced by $s^\text{c,no app}_t$: In this case, it doesn't matter whether or not the infector has the app. The equation simplifies to a form analogous to Eq.\ \eqref{Eq: time evolution equation}, namely
\begin{equation}\label{Eq: time evolution equation without app}
    F_{t,\text{no app}}^\text{A,c}(\rho) = 
s^\text{c,no app}_t\,\sum_g\int_{(0, +\infty)}
    \frac{F^\text{T}_{t-\tau,g}(\rho + \tau) - F^\text{T}_{t-\tau,g}(\tau)}
    {1-\xi_{t-\tau} F_{t-\tau,g}^\text{T}(\tau)}\,
    \mathrm{d}F^{\sigma,g}_t(\tau)\,.
\end{equation}
Again, we refer to the Appendix for a greater mathematical rigor: The last two equations are derived in greater detail in \S\ref{A: scenario with app}.
\section{Scenarios and calculations}\label{Calculations}

In this section, we use the models introduced in \S\ref{Mathematical model} and \S\ref{Modifications with app} to numerically compute the suppression of $R_t$ due to isolation measures in certain scenarios.

The results reported here, as well as new custom calculations, can be obtained by cloning the public Python repository \cite{mm2021epidemics}.

\subsection{General considerations} \label{calculations: general}

Some inputs of the algorithm developed are parameters or distributions describing the features of the epidemic under consideration. In this section, we focus on COVID-19, and we make the following assumptions, taking all the epidemic data from \cite{FerrEtAl2020Quant} (Table 1, in particular) for convenience:
\begin{itemize}
\item The incubation period $\tau^S$ is distributed according to a log-normal distribution:
\begin{equation*}
F^\text{S}(\tau) = F_{N_{0,1}}\left((\log(\tau) - \mu)\sigma\right)\,
\end{equation*}
where $F_{N_{0,1}}$ denotes the CDF of the standard normal distribution. The parameters $\mu=1.64$ and $\sigma=0.36$ used here imply that the mean incubation period is $\simeq 5.5$ days.

\item The default infectiousness distribution $\beta^0_t$ is assumed to depend on the absolute time $t$ only via a global factor, so that
$$\beta^0_t(\tau)=R^0_t\rho^0(\tau)\,,$$
where $\rho^0$ (which also represents the default generation time distribution) integrates to 1. It is described by a Weibull distribution with mean 5.00 and variance 3.61:
$$\rho^0(\tau) = \frac{k}{\lambda}\left(\frac{\tau}{\lambda}\right)^{k - 1}  e^{-\left(\tau/\lambda\right) ^ k}$$
with $k=2.855$, $\lambda=5.611$.
\item We simplify the severity of symptoms by considering only two levels of severity: $g=\text{sym}$ and $g=\text{asy}$, respectively for symptomatic and asymptomatic individuals. We take asymptomatic individuals as 40\%, and we assume that they account for 5\% of $R_t$.\footnote{Note that this does not include pre-symptomatic transmission, which is taken into account within the group $g=\text{sym}$.} In formulae, this means that the input parameters of our model are
\begin{equation*}
\begin{array}{c}
p_\text{sym}=0.6\,,\quad p_\text{asy}=0.4\,,\\
\beta^0_{t, \text{sym}}=\frac{0.95}{0.6}R^0_t\rho^0\,,\quad
\beta^0_{t, \text{asy}}=\frac{0.05}{0.4}R^0_t\rho^0\,.
\end{array}
\end{equation*}
\end{itemize}

All these assumptions hold throughout the whole section except for \S\ref{Dependency on epidemic data}, where we check how the results change using different epidemic data. The other parameters of the model, describing the isolation measures, are selected later. 

As Key Performance Indicators (KPIs) describing the effectiveness of the isolation measures, we look at the reduction of $R_t$ compared to the value $R_t^0$ it would take in the absence of measures. We call \emph{effectiveness} of the isolation measures the relative suppression of $R_t^0$:
\begin{equation}%\label{Eq: effectiveness from R}
\text{Eff}_t := 1 - \frac{R_t}{R^0_t}\,.
\end{equation} 
Thus, $\text{Eff}_t=0$ indicates that there is no effect on $R^0_t$, while $\text{Eff}_t=1$ describes a complete suppression of the contagion. We will see in \S\ref{Dependency on epidemic data} that the dependency of $\text{Eff}_t$ on the default reproduction number $R^0_t$ is very weak.\footnote{This dependency is due to higher order effects: A higher $R^0_t$ means that the distribution of $\tau^\sigma_t$ is more concentrated on small values, and hence the most recent testing time distributions have a greater weight in the time-evolution equation \eqref{Eq: time evolution equation}.} As such, attempts to model a realistic profile for $R^0_t$ have little relevance to our computations, as any choice of $R^0_t$ leads to almost the same $\text{Eff}_t$. Thus, in the rest of this section we simply take
$$R^0_t=1\,.$$

Another useful KPI is the probability that an individual infected at a certain time $t$ is eventually found to be positive, namely the limit
$$F^\text{T}_t(\infty) := \lim_{\tau\to+\infty}F^\text{T}_t(\tau)\,.$$

In the remainder of this section, we report the results of some selected calculations, considering both the ``homogeneous'' scenario of \S\ref{Mathematical model} and, in greater detail, the scenario of \S\ref{Modifications with app}, in which an app for epidemic control is used. First, we study how the above KPIs evolve in time for certain input parameter choices. Then, we focus on the limits for $t\to+\infty$ of these KPIs, i.e., their ``stable'' values after a sufficient number of iterations, to study how these vary when we change certain input parameters, leaving the others fixed.

\subsection{Suppression of \texorpdfstring{$R_t$}{Rt} with homogeneous isolation measures}

First, we perform some calculations in the setting of \S\ref{Mathematical model}, where the isolation measures are ``homogeneous'' within the whole population. We recall the parameters that describe this situation, some of which remain fixed in all the calculations:\footnote{In all the examples these parameters are constant in time. Hence we remove the subscript $t$ from them.}

\begin{center}
 \begin{tabular}{||p{1.8cm} | p{5.5cm} | p{3cm}||} 
 \hline
 Parameter & Meaning & Value\\ [0.5ex] 
 \hline\hline
 $s^\text{s}_\text{sym}$ & Probability that a symptomatic infected individual is notified of the infection because of their symptoms & Not fixed \\ 
 \hline
 $s^\text{s}_\text{asy}$ & As above, but for the asymptomatic & 0\\ 
 \hline
 $\xi$ & Probability that someone testing positive self-isolates & Not fixed \\ 
 \hline
 $\Delta^{\text{A}\to\text{T}}$ & Time from notification to positive testing & Constant distribution, whose value is not fixed at this moment\\
 \hline
 $t_0$ & Time at which isolation measures begin & 0 \\
 \hline
\end{tabular}
\end{center}
Note that assuming that $\Delta^{\text{A}\to\text{T}}$ is a constant random variable means that we are modeling that all individuals notified of the risk test positive, and take the same time to do so. Although unrealistic, this assumption makes little difference to the results. It is made here for simplicity, although it can be easily changed by using a more realistic $\Delta^{\text{A}\to\text{T}}$, when this datum is available.

\subsubsection{Time evolution with isolation due to both symptoms and contact-tracing}\label{Time ev homog sympt+ct}

We now choose the following parameters, describing an optimistic situation, with reasonable efficiencies in spotting infected individuals:
\begin{center}
 \begin{tabular}{||l | c||} 
 \hline
 Parameter & Value\\ [0.5ex] 
 \hline\hline
 $s^\text{s}_\text{sym}$ & 0.5 \\ 
 \hline
 $s^\text{c}$ &  0.7\\ 
  \hline
 $\xi$ & 0.9 \\ 
 \hline
 $\Delta^{\text{A}\to\text{T}}$ & 2 \\ 
 \hline
\end{tabular}
\end{center}

\begin{figure}[htb]
\begin{center}
\includegraphics[width=0.7\textwidth]{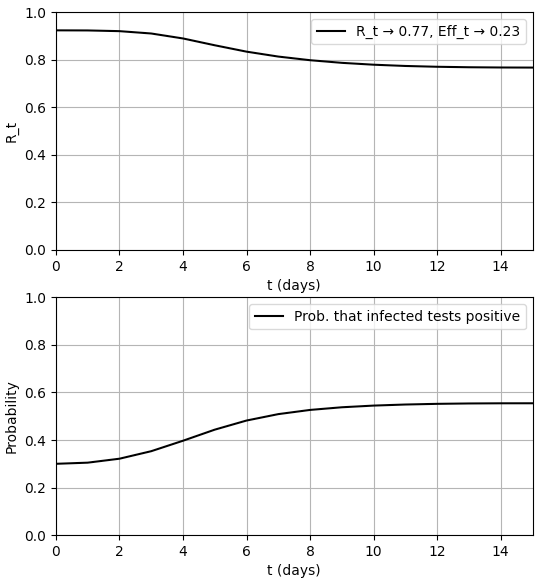}    
\end{center}
\caption{$R_t$ evolution in the homogeneous model, in an optimistic scenario}
\label{Fig: homogeneous}
\end{figure}

The results are shown in Fig.\ \ref{Fig: homogeneous}. Note that immediately at $t=0$ $R_t$ drops to around 0.92, as half of the symptomatic individuals are notified as soon as they show some symptoms, and they then infect a reduced number of people. Subsequently, $R_t$ continues to decrease due to contact-tracing, quickly approaching its limit value $R_\infty\simeq 0.77$ (i.e.,\ 84\% of the value it would have had with isolation due to symptoms only).

\subsubsection{Dependency on testing timeliness}

We now focus on the limit value $\text{Eff}_\infty$, investigating its dependency on the time $\Delta^{\text{A}\to\text{T}}$ from a notification to the positive result of the test (recall that we are assuming that $\Delta^{\text{A}\to\text{T}}$ is a constant random variable).
 
Like in \S\ref{Time ev homog sympt+ct}, the other parameters are fixed as follows:
\begin{center}
 \begin{tabular}{||l | c||} 
 \hline
 Parameter & Value\\ [0.5ex] 
 \hline\hline
 $ s^\text{s}_\text{sym} $ & 0.5 \\ 
 \hline
 $ s^\text{c} $ &  0.7\\ 
  \hline
 $\xi$ & 0.9 \\ 
 \hline
\end{tabular}
\end{center}

The result is plotted in Fig.\ \ref{fig:timeliness_span}. The effectiveness of the isolation measures improves dramatically with the ability to test (and then isolate) infected individuals as soon as possible after their notification of possible infection.

\begin{figure}[htb]
\begin{center}
\includegraphics[width=0.5\textwidth]{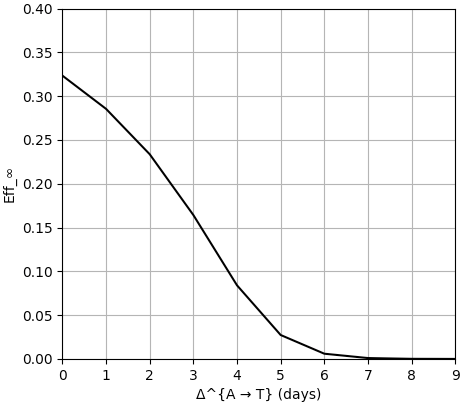}
\end{center}
\caption{$\text{Eff}_\infty$ as a function of the time $\Delta^{\text{A}\to\text{T}}$ from notification to positive testing.}
\label{fig:timeliness_span}
\end{figure}

\subsubsection{Dependency on the epidemic data used}\label{Dependency on epidemic data}

In this subsection we briefly explore what happens if we change some of the data describing the epidemic, that were introduced in \S\ref{calculations: general} and used elsewhere in this section. This is done to see how $\text{Eff}_\infty$ depends on these data. The other parameters, describing the isolation measures, are fixed as usual:
\begin{center}
 \begin{tabular}{||l | c||} 
 \hline
 Parameter & Value\\ [0.5ex] 
 \hline\hline
 $s^\text{s}_\text{sym}$ & 0.5 \\ 
 \hline
 $s^\text{c}$ &  0.7\\ 
  \hline
 $\xi$ & 0.9 \\ 
 \hline
 $\Delta^{\text{A}\to\text{T}}$ & 2 \\ 
 \hline
\end{tabular}
\end{center}

First, we let the fraction $p_\textup{sym}$ of symptomatic individuals vary, along with their contribution to $R^0_t$---let us denote it here by $\kappa$---that is elsewhere taken as $\kappa=0.95$. Recall that
$$
R^0_{t,\textup{sym}}=\frac{\kappa}{p_\textup{sym}}R^0_t\,,\quad
R^0_{t,\textup{asy}}=\frac{1-\kappa}{1-p_\textup{sym}}R^0_t\,.
$$
The value of $\text{Eff}_\infty$ for a few choices of $p_\textup{sym}$ and $\kappa$ is plotted in Fig.\ \ref{fig:symptomatics_dependency}, where it is apparent how the result is robust with respect to changes in these input data. Note that if we fix $p_\textup{sym}$ and let $\kappa$ vary, then the two components of $R^0_t$ (and hence those of $R_t$) are linearly rescaled, meaning that $\text{Eff}_\infty$ also changes linearly.

\begin{figure}[htb]
\begin{center}
\includegraphics[width=0.7\textwidth]{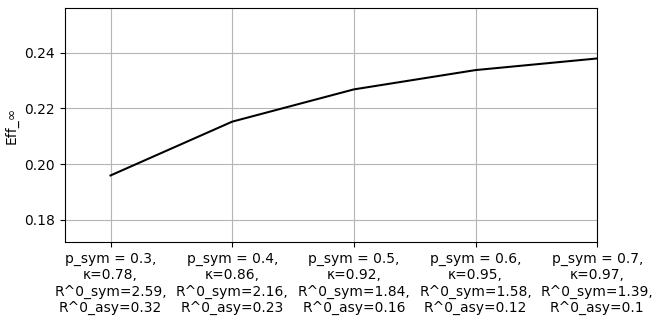}
\includegraphics[width=0.7\textwidth]{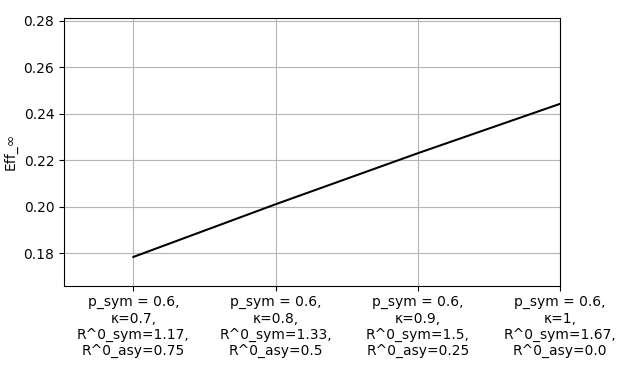}
\end{center}
\caption{$\text{Eff}_\infty$ for some values of $p_\textup{sym}$ and $\kappa$.}
\label{fig:symptomatics_dependency}
\end{figure}

Second, we fix $p_\textup{sym}=0.6$ and $\kappa=0.95$ as usual, and we modify instead the density $\rho^0$ of the default generation time, by replacing it with
$$\rho_f^0(\tau)=\frac{1}{f}\rho^0(\tau/f)$$
for $f>0$. Note that this implies that the expected value of the default generation time (denoted here by $\tau^{0,\textup{C}}$) is multiplied by $f$:
% $$\int_{[0,+\infty)}\tau\rho_f^0(\tau)\mathrm{d}\tau = f\,\int_{[0,+\infty)}\tau\rho^0(\tau)\mathrm{d}\tau = f(5\textup{ days})\,.$$
$$\mathbb{E}(\tau^{0,\textup{C}}) = f(5\textup{ days})\,.$$

Fig.\ \ref{fig:generation_time_dependency} depicts the relation between $\mathbb{E}(\tau^{0,\textup{C}})$ and $\text{Eff}_\infty$, as $f$ varies. As expected, the isolation measures become more effective as the time taken by the infection to be transmitted increases.

\begin{figure}[htb]
\begin{center}
\includegraphics[width=0.5\textwidth]{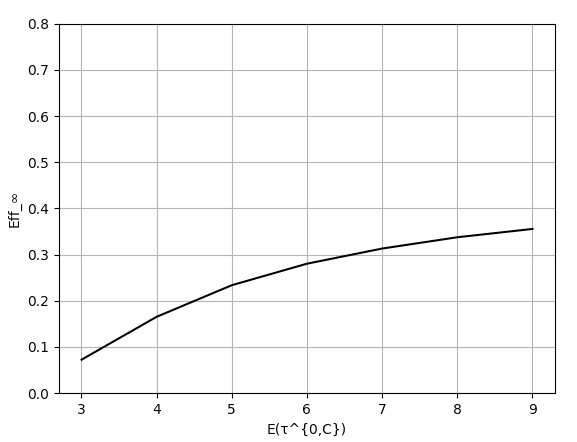}
\end{center}
\caption{$\text{Eff}_\infty$ for some rescalings of the distribution of $\tau^\textup{0,C}$.}
\label{fig:generation_time_dependency}
\end{figure}

Finally, Fig.\ \ref{fig:R0_dependency} shows how $\text{Eff}_\infty$ changes slightly as we change the value of $R^0_t$ (for $t\geq 0$).

\begin{figure}[htb]
\begin{center}
\includegraphics[width=0.5\textwidth]{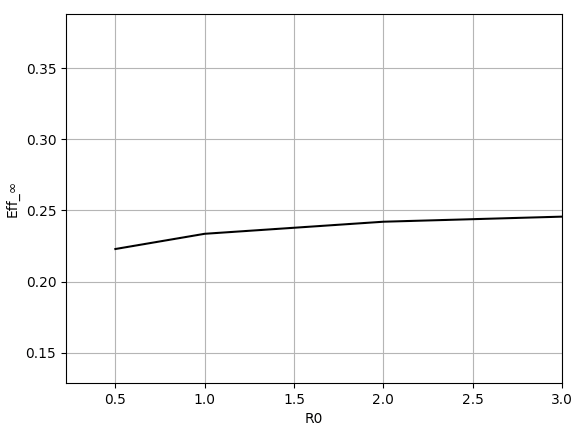}
\end{center}
\caption{$\text{Eff}_\infty$ for some values of $R^0_t$, for $t\geq 0$.}
\label{fig:R0_dependency}
\end{figure}

\subsection{Suppression of \texorpdfstring{$R_t$}{Rt} in the case of app usage}

Now, we focus on applying the model of \S\ref{Modifications with app} to study how $R_t$ is suppressed when a fraction of the population uses an app for epidemic control.

In this case, we summarize the input parameters in the following table, fixing some of them to the given values for the rest of the section (unless explicitly mentioned).

\begin{center}
 \begin{tabular}{||p{1.8cm} | p{5.5cm} | p{3cm}||} 
 \hline
 Parameter & Meaning & Value\\ [0.5ex] 
 \hline\hline
 $s^\text{s,app}_\text{sym}$ & Probability that a symptomatic infected individual using the app is notified of the infection because of their symptoms & Not fixed \\ 
 \hline
 $s^\text{s,no app}_\text{sym}$ & As above, but for individuals without the app & 0.2 \\ 
 \hline
 $s^\text{s,app}_\text{asy}$, $s^\text{s,no app}_\text{asy}$ & As with the two parameters above, but for asymptomatic individuals & 0, 0\\ 
 \hline
  $s^\text{c,app}$ & Probability that an infected individual with the app is notified of the infection because of their source having tested positive & Not fixed\\
 \hline
 $s^\text{c,no app}$ & Probability that an infected individual without the app is notified of the infection because of their source having tested positive & 0.2\\
 \hline
 $\xi$ & Probability that someone testing positive self-isolates & Not fixed \\ 
 \hline
 $\Delta^{\text{A}\to\text{T}, \text{app}}$, $\Delta^{\text{A}\to\text{T}, \text{no app}}$ & Time from notification to positive testing for people with and without the app, respectively & Constant distributions, whose values are not fixed at this moment\\
 \hline
 $\epsilon_{t,\text{app}}$ & Fraction of the population adopting the app at time $t$ & Not fixed \\
 \hline
 $t_0$ & Time at which isolation measures begin & 0 \\
 \hline
\end{tabular}
\end{center}

\subsubsection{Time evolution in an optimistic scenario}\label{Time evolution optimistic scenario}

We start with an optimistic scenario, where the app is effective at recognizing infected individuals from symptoms and contact-tracing information. The internal predictive models that estimate the probability of an individual being infected have high efficiencies (a situation likely bound to the possibility of training the predictive models on real data, in practice). The app is adopted by a large fraction (60\%) of the population, and is trusted, so that most of the people notified take a test and self-isolate. The app also helps a notified individual to get tested more quickly.\footnote{Studies such as \cite{li2020early} report that the time from symptom onset to testing through the ``conventional'' channels (health care system) is in the order of several days. An app is expected to have substantial chances to improve this performance, being a prompt-instrument by construction (for example, when compared with the friction of calling a doctor, inserting symptom descriptions into the app is likely easier), so that
$$
\Delta^{\text{A}\to\text{T}, \text{app}} < \Delta^{\text{A}\to\text{T}, \text{no app}}\,.
$$}

\begin{center}
 \begin{tabular}{||l | c||} 
 \hline
 Parameter & Value\\ [0.5ex] 
 \hline\hline
 $s^\text{s,app}_\text{sym}$ & 0.8 \\ 
 \hline
 $s^\text{c,app}$ &  0.8\\ 
 \hline
 $\xi$ & 0.9 \\ 
 \hline
 $\epsilon_\text{app} $ & 0.6 \\
 \hline
 $\Delta^{\text{A}\to\text{T}, \text{app}}$ & 2 \\ 
 \hline
 $\Delta^{\text{A}\to\text{T}, \text{no app}}$ & 4 \\ 
 \hline
\end{tabular}
\end{center}

As we start from $R^0_t = 1$, we reach a limit value of $R_\infty \simeq 0.84$ for an effectiveness of $0.16$. Note also that $R_{\infty,\text{app}}\simeq 0.75$, while $R_{\infty,\text{no app}}\simeq 0.97$. The time evolution of $R_t$, along with the other main quantities of interest, is shown in Fig.\ \ref{fig:good_app}.

\begin{figure}[htb]
\includegraphics[width=\textwidth]{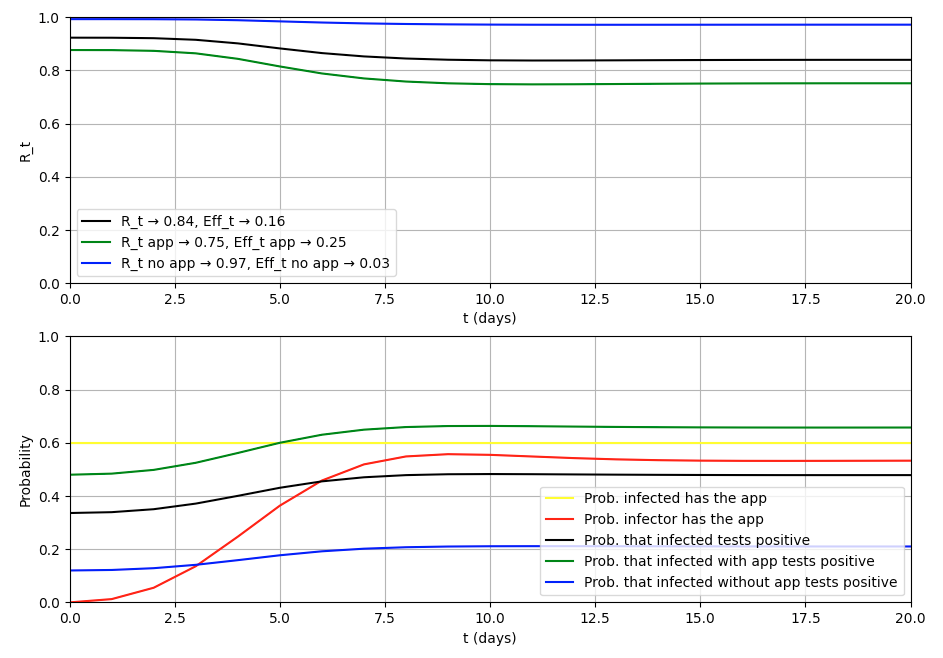}
\caption{KPIs evolution in the optimistic scenario.}
\label{fig:good_app}
\end{figure}

\subsubsection{Time evolution in a pessimistic scenario}

We now run an analogous computation in a ``pessimistic'' scenario. The app can only recognize infected individuals from contact-tracing information, and not from symptoms ($s^\text{s,app}_\text{sym}$ consequently defaults to the no-app value). In addition, we assume a low efficiency $s^\text{c,app} = 0.5$, perhaps due to poor predictive models. Also, only 70\% of those testing positive self-isolate.

\begin{center}
 \begin{tabular}{||l | c||} 
 \hline
 Parameter & Value\\ [0.5ex] 
 \hline\hline
 $s^\text{s,app}_\text{sym}$ & 0.2 \\ 
 \hline
 $s^\text{c,app}$ &  0.5\\ 
 \hline
  $\xi$ & 0.7 \\ 
 \hline
 $\epsilon_\text{app} $  & 0.6 \\ 
 \hline
 $\Delta^{\text{A}\to\text{T}, \text{app}}$ & 2 \\ 
 \hline
 $\Delta^{\text{A}\to\text{T}, \text{no app}}$ & 4 \\ 
 \hline
\end{tabular}
\end{center}

Even with a high app adoption rate (60\% of the population), the effectiveness drops dramatically. We get $R_\infty \simeq 0.96$ and $\text{Eff}_\infty\simeq 0.06$. Most notably, the app does not change things much with respect to ``standard'' isolation measures: $R_{\infty,\text{app}} \simeq 0.95$ and $R_{\infty,\text{no app}}\simeq 0.99$.

\subsubsection{Time evolution in the case of gradual adoption of the app}

Now, we study the evolution of $R_t$ in a scenario whereby the fraction $\epsilon_{t,\text{app}}$ of people using the app is not constant, but increasing in a linear fashion until it reaches 60\% in 30 days:
$$\epsilon_{t,\text{app}}=0.6t/30\text{ for }0\leq t<30,\ \epsilon_{t,\text{app}}=0.6\text{ for }t\geq 30\,.$$
The other parameters are chosen as in the optimistic scenario of \S\ref{Time evolution optimistic scenario}:

\begin{center}
 \begin{tabular}{||l | c||} 
 \hline
 Parameter & Value\\ [0.5ex] 
 \hline\hline
 $s^\text{s,app}_\text{sym}$ & 0.8 \\ 
 \hline
 $s^\text{c,app}$ &  0.8\\ 
 \hline
 $\xi$ & 0.9 \\ 
 \hline
 $\Delta^{\text{A}\to\text{T}, \text{app}}$ & 2 \\ 
 \hline
 $\Delta^{\text{A}\to\text{T}, \text{no app}}$ & 4 \\ 
 \hline
\end{tabular}
\end{center}

As shown in Fig.\ \ref{Fig: gradual app adoption}, $R_t$ decreases until stabilizing again to the same value obtained in \S\ref{Time evolution optimistic scenario}, although it takes more time to do so. The limit values of the KPIs are not changed by a gradual adoption of the app, compared with a prompt adoption.

\begin{figure}[htb]
\includegraphics[width=\textwidth]{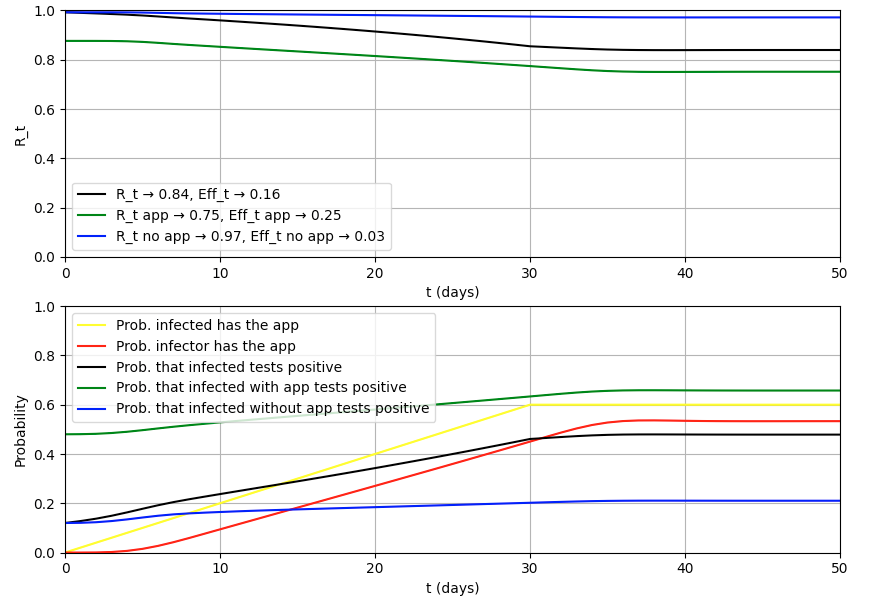}
\caption{KPIs evolution in case of gradual adoption of the app.}
\label{Fig: gradual app adoption}
\end{figure}

\subsubsection{Dependency of effectiveness on the efficiencies \texorpdfstring{$s^\text{s}$ and $s^\text{c}$}{ss and sc}}

We now focus on the study of how the limit values of the KPIs change when we vary certain parameters, starting with the app efficiencies $s^\text{s,app}_\text{sym}$ and $s^\text{c,app}$. In Fig.\ \ref{fig:efficiency_span}, we plot $\text{Eff}_\infty$ as a function of these two parameters, while the others are fixed to the following values:
\begin{center}
 \begin{tabular}{||l | c||} 
 \hline
 Parameter & Value\\ [0.5ex] 
 \hline\hline
 $\xi$ & 0.9 \\ 
 \hline
 $\epsilon_\text{app} $ & 0.6 \\
 \hline
 $\Delta^{\text{A}\to\text{T}, \text{app}}$ & 2 \\ 
 \hline
 $\Delta^{\text{A}\to\text{T}, \text{no app}}$ & 4 \\ 
 \hline
\end{tabular}
\end{center}

\begin{figure}[htb]
\begin{center}
\includegraphics[width=0.5\textwidth]{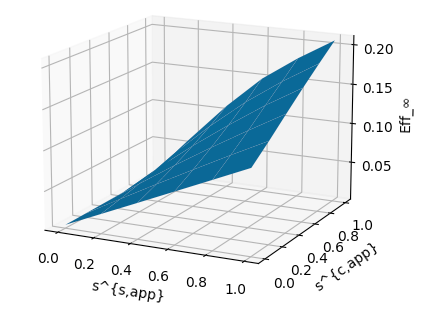}
\end{center}
\caption{$\text{Eff}_\infty$ as a function of the efficiencies $s^\text{s,app}_\text{sym}$ and  $s^\text{c,app}$.}
\label{fig:efficiency_span}
\end{figure}

\subsubsection{Dependency on the app adoption}

\begin{figure}[htb]
\begin{center}
\includegraphics[width=0.5\textwidth]{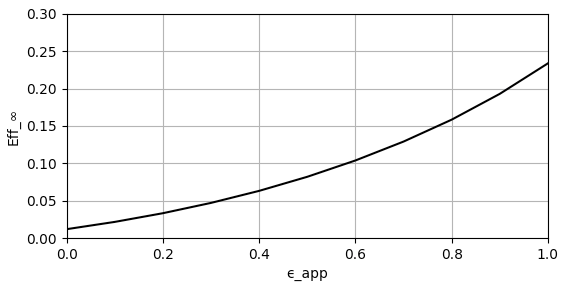}
\end{center}
\caption{$\text{Eff}_\infty$ as a function of app adoption $\epsilon_\text{app}$.}
\label{fig:adoption_span}
\end{figure}

In Fig.\ \ref{fig:adoption_span}, we can observe the dependency of the effectiveness $\text{Eff}_\infty$ on the share $\epsilon_\text{app}$ of the population using the app. The remaining parameters are fixed to these values:
\begin{center}
 \begin{tabular}{||l | c||} 
 \hline
 Parameter & value\\ [0.5ex] 
 \hline\hline
 $s^\text{s,app}_\text{sym}$ & 0.5 \\ 
 \hline
  $s^\text{c,app}$ &  0.7\\ 
 \hline
  $\xi$ & 0.9 \\ 
 \hline
 $\Delta^{\text{A}\to\text{T}, \text{app}}$ & 2 \\ 
 \hline
 $\Delta^{\text{A}\to\text{T}, \text{no app}}$ & 4 \\ 
 \hline
\end{tabular}
\end{center}

\appendix
\section{Appendix: formalism and detailed derivations of mathematical results}

The goal of this Appendix is to introduce a mathematical framework in which the hypotheses of our model can be formulated precisely and their consequences proven rigorously. In particular, we will derive the formula \eqref{Eq: beta_t,g corr} describing the suppression of $R_t$ due to the testing and isolation policies, and the time evolution equation \eqref{Eq: time evolution equation}. Some more details about the two-component scenario of \S\ref{Modifications with app} are added.

\subsection{Modeling a deterministic epidemics with a probability space}

In this subsection, we define the fundamental components of our framework. $\Omega$ denotes the set of all the individuals infected during the epidemic. It is endowed with two functions. First,
$$t^\text{I}:\Omega\to[0,+\infty)$$
associates to each individual $\omega\in\Omega$ the absolute time of their infection. Note that we take $0$ as the initial time of the epidemics. $t^I$ partitions $\Omega$ into a foliation
$$\Omega = \cup_{t\in\mathbb{R}}\Omega_t
 := \cup_{t\in\mathbb{R}}\{\omega\in\Omega\ |\ t^\text{I}(\omega)=t\}\,.$$
We also write
$$
\Omega_{>0}:=\cup_{t\in\mathbb{R}^+}\Omega_t
$$
for the set of individuals infected at a positive time.

Second, for any individual $\omega\in\Omega_{>0}$ we denote by $\sigma(\omega)\in\Omega$ their infector. This defines a map
$$\sigma:\Omega_{>0}\to\Omega\,.$$
It is natural to assume that, for all $\omega\in\Omega_{>0}$,
$$
t^\text{I}(\sigma(\omega)) < t^\text{I}(\omega)\,. 
$$

\begin{figure}[htb]
\begin{center}
\includegraphics[width=0.5\textwidth]{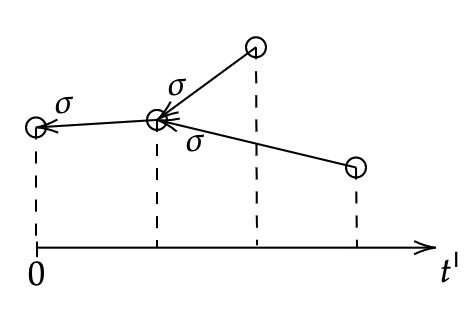}
\end{center}
\caption{Schematic representation of $\Omega$, $t^\text{I}$, and $\sigma$. Dots represent individuals $\omega\in\Omega$.}
\label{Fig: omega}
\end{figure}

As is the case in the rest of the paper, we will consider other quantities referring to infected individuals and study the mathematical relations between them, with special attention to their \emph{average} properties over all individuals infected at a given time $t$. Hence, we will study functions defined over $\Omega$, and to talk about their averages over each set $\Omega_t$ we will introduce a probability measure $P$ on $\Omega$. Since $\Omega$ is finite and we want to weight all individuals equally, $P$ is the uniform discrete probability measure:
$$
P(E) = \frac{\#E}{\#\Omega}\,,
$$
where $\#E$ denotes the cardinality of a set $E\subset\Omega$. Relevant quantities then become random variables, and we are interested in studying their distributions. This always reduces to solving certain counting problems, and introducing $P$ is largely a way to conveniently write formulas using the language of probability theory. Let us stress that our methodology does not involve any simulation of random processes: The history of the epidemic is completely determined by the triple $(\Omega, t^\text{I},\sigma)$, and our study of its evolution consists of writing deterministic relations that express random variables restricted to a time slice $\Omega_t$ in terms of random variables restricted to slices $\Omega_{t'}$, for times $t'<t$.

The probability measure $P$ ``disintegrates'' along $t_I$, giving a uniform probability measure $P_t$ on each $\Omega_t$. Also, for any random variable $X:\Omega\to\mathbb{R}$ we will denote by
$$
X_t := X{\restriction_{\Omega_t}}
$$
its restriction to $\Omega_t$. We will be mostly interested in studying the distributions of such restrictions of random variables. In considering quantities like the expected value
$$
\mathbb{E}_{P_t}(X_t) = \int_{\Omega_t}X\,\mathrm{d}P_t = \frac{1}{\#\Omega_t}\sum_{\omega\in\Omega_t}X(\omega)
$$
of $X_t$, we will always make the relevant probability measure $P_t$ explicit to avoid confusion.

%%%

\subsection{Infector-infectee pairs, generation time and the reproduction number}

Let us introduce some additional notation, for future convenience: First, we define the set
$$
\tilde\Omega = \{(\omega',\omega)\in\Omega\times\Omega_{>0}\ |\ \sigma(\omega) = \omega'\}
$$
of infector-infectee pairs. This is the graph of the map $\sigma$, which thus determines a bijection $\Omega_{>0}\to\tilde\Omega$. We also consider two functions describing the generation time: a function $\tau^\text{C}:\tilde\Omega\to\mathbb{R}^+$, given by
$$
\tau^\text{C}(\omega',\omega) := t^\text{I}(\omega) - t^\text{I}(\omega')\,,
$$
and a function $\tau^\sigma:\Omega_{>0}\to\mathbb{R}^+$ given by
$$
\tau^\sigma(\omega) := \tau^\text{C}(\sigma(\omega))\,.
$$

Consider, for all $\tau\in\mathbb{R}^+$, the random variable $n^\tau\!:\Omega\to\mathbb{N}$ such that $n^\tau(\omega)$ is the number of people infected by $\omega$ at $\omega$'s infectious age $\tau$ (that is, at the absolute time $t^\text{I}(\omega) + \tau$):
$$
n^\tau(\omega) := \#\{\omega'\in\Omega_{>0}\ |\ \sigma(\omega')=\omega\,,\  \tau^\sigma(\omega')=\tau\}\,.
$$
Also, for $\tau\in[0,+\infty]$, we denote by $N^\tau\!:\Omega\to\mathbb{N}$ the random variable that counts all individuals infected \emph{within} the infectious age $\tau$:
$$
N^\tau(\omega) := \sum_{\tau'\leq \tau}n^{\tau'}(\omega) = \#\{\omega'\in\Omega_{>0}\ |\ \sigma(\omega')=\omega\,,\  \tau^\sigma(\omega')\leq\tau\}\,.
$$
In particular, $N^\infty(\omega)$ is the total number of people infected by $\omega$. The average values of these variables are key indicators of the speed of propagation of the epidemics. In particular, restricting to an absolute time $t$, we let
$$
R_t := \mathbb{E}_{P_t}(N_t^\infty)
$$
be the \emph{effective reproduction number} at $t$. Note that averaging on \emph{all} infected individuals simply gives $\mathbb{E}_P(N^\infty)=1$. We also consider the average values of $N_t^\tau$, for finite $\tau$:
$$
B_t(\tau) := \mathbb{E}_{P_t}(N_t^\tau)\,.
$$
Notice that $B_t$ is an improper CDF supported on $\mathbb{R}^+$, and because of the finiteness of $\Omega_t$, it is a sum of finitely many step functions. It represents the \emph{cumulative infectiousness} of individuals infected at $t$, and tends to $R_t$ for $\tau\to+\infty$.

For practical reasons, it is common to tacitly consider a \emph{continuum limit} in which each $\#\Omega_t$ tends to infinity and all random variables become continuous. $B_t$ is then approximated by a smooth function, whose derivative is denoted by $\beta_t$ (as is the case in the rest of the paper):
$$
B_t(\tau) \simeq \int_0^\tau\beta_t(\tau')\,\mathrm{d}\tau'\,.
$$
However, in this Appendix we always work in the discrete setting discussed so far, and then we consider the continuum limit only to get formulas for $\beta_t$, for consistency with the standard terminology and notation. Using the formalism of measure theory, or simply writing the relations between random variables in terms of their CDFs, allows us to treat both the discrete scenario and the continuum limit in a unified notation.

%%%

\subsection{The suppression formula}\label{A: suppression formula}

Here, we discuss in greater mathematical detail the content of \S\ref{Rt suppression model}. In particular, we will derive the suppression formula (Eq.\ \eqref{Eq: beta_t,g corr}) relating the reproduction number $R_t$ and the distribution of the random variable $\tau^\text{T}\!:\Omega\to[0,+\infty]$ describing the infectious age (possibly infinite) at which each individual is tested positive.

To do this, we introduce two additional random variables $n^{0,\tau},N^{0,\tau}:\Omega\to\mathbb{N}$ which are analogues to $n^\tau$ and $N^\tau$, but which instead count the number of individuals that each $\omega\in\Omega$ \emph{would have infected without isolation measures}. Similarly, we denote by
$$
B_t^0(\tau) := \mathbb{E}_{P_t}(N_t^{0,\tau})\,,\quad
R^0_t := \mathbb{E}_{P_t}(N_t^{0,\infty})\,,
$$
and $\beta^0_t$ the analogues of $B_t$, $R_t$ and $\beta_t$ in the absence of isolation measures.

Recall that we assumed the average number of people infected by each individual $\omega$ is reduced by a factor $1-\xi_t$ (possibly depending on the infection time $t := t^\text{I}(\omega)$) at times $\tau$ greater or equal than the testing time $\tau^\text{T}(\omega)$. This is encoded by the following relation between the expected values of $n_t^\tau$ and $n_t^{0,\tau}$ conditioned by $\tau_t^\text{T}$: For all $\tau,\rho\in[0,+\infty]$, we postulate that
\begin{equation}\label{Eq: suppr hyp}
\begin{array}{rcl}
    \mathbb{E}_{P_t}(n_t^\tau|\tau_t^\text{T}=\rho) 
    &=&  
    (1-\xi_t\,\delta_{\tau\geq\rho})\,\mathbb{E}_{P_t}(n_t^{0,\tau}|\tau_t^\text{T}=\rho)\\
    &=& 
    \left\{\begin{matrix}
    \mathbb{E}_{P_t}(n_t^{0,\tau}|\tau_t^\text{T}=\rho) & \text{for }\tau<\rho\\ 
    (1-\xi_t)\,\mathbb{E}_{P_t}(n_t^{0,\tau}|\tau_t^\text{T}=\rho) & \text{for }\tau\geq\rho
    \end{matrix}\right.\,,
\end{array}
\end{equation}
where $\delta_{\tau\geq\rho}:=\chi_{[\rho, +\infty]}(\tau)$ is $1$ when $\tau\geq\rho$ and $0$ otherwise. Now we would like to remove the conditioning on $\tau_t^\text{T}$ from the expected values of $n_t^{0,\tau}$ to get an expression in terms of known quantities only. If, for simplicity, we supposed that $\tau^\text{T}$ and $n_t^{0,\tau}$ are independent, then Eq.\ \eqref{Eq: suppr hyp} would reduce to $(1-\xi\,\delta_{\tau\geq\rho})\,\mathbb{E}_{P_t}(n_t^{0,\tau})$.
But, as discussed in \S\ref{Rt suppression model}, this is not a realistic hypothesis. Instead, we only assume that $\tau^\text{T}$ and $n_t^{0,\tau}$ are independent when restricted to individuals having the same \emph{severity of illness}, which we describe through a random variable
$$G:\Omega\to\mathbb{R}\,.$$
In other words, we take $\tau_t^\text{T}$ and $n_t^{0,\tau}$ to be \emph{conditionally} independent with respect to $G$.\footnote{Remember that in this work, the joint distribution of $n^{0,\tau}_t$ and $G$ is assumed known, and we want to study how our assumptions on the isolation measures determine the distribution of $n_t^\tau$, and hence $R_t$.} We assume now that the suppression formula applies equally to individuals of all degrees of severity:
\begin{equation}\label{Eq: suppr hyp by sev}
    \mathbb{E}_{P_{t,g}}(n_{t,g}^\tau|\tau_{t,g}^\text{T}=\rho) 
    = (1-\xi_t\,\delta_{\tau\geq\rho})\,\mathbb{E}_{P_{t,g}}(n_{t,g}^{0,\tau}|\tau_{t,g}^\text{T}=\rho)\,,
\end{equation}
having introduced
$$\Omega_{t,g}:=\{\omega\in\Omega_t\ |\ G(\omega)=g\}$$
and the obvious notation for restrictions of random variables to $\Omega_{t,g}$ and for the uniform probability measure $P_{t,g}$ on it. The assumption of conditional independence implies
\begin{equation}\label{Eq: suppr hyp 2}
    \mathbb{E}_{P_{t,g}}(n_{t,g}^\tau|\tau_{t,g}^\text{T}=\rho) = 
(1-\xi_t\,\delta_{\tau\geq\rho})\,\mathbb{E}_{P_{t,g}}(n_{t,g}^{0,\tau})\,.
\end{equation}
Summing over $\rho\in[0,+\infty]$ we find
$$
\begin{array}{rcl}
    \mathbb{E}_{P_{t,g}}(n_{t,g}^\tau) &=&
       \sum_{\rho\in[0,+\infty]}
       P_{t,g}(\tau_{t,g}^\text{T}=\rho)\,
       \mathbb{E}_{P_{t,g}}(n_{t,g}^\tau|\tau_{t,g}^\text{T}=\rho)\\
     &=& 
       \mathbb{E}_{P_{t,g}}(n_{t,g}^{0,\tau})
       \sum_{\rho\in[0,+\infty]}(1-\xi_t\,\delta_{\tau\geq\rho})\,P_{t,g}(\tau_{t,g}^\text{T}=\rho)\\
      &=& 
    \mathbb{E}_{P_{t,g}}(n_{t,g}^{0,\tau})(1-\xi_t F_{t,g}^\text{T}(\tau))\,,
\end{array}
$$
having denoted the improper CDF of $\tau_{t,g}^\text{T}$ by $F_{t,g}^\text{T}$, as usual.

Finally, to average over all $\Omega_t$, we simply notice that
$$
    \mathbb{E}_{P_t}(n_t^\tau) =
    \sum_g p_{t,g}\,
    \mathbb{E}_{P_{t,g}}(n_{t,g}^\tau)\,,
$$
where $p_{t,g}:=P_t(G_t=g)$.\footnote{Note that, for simplicity, in the rest of the paper we took the distribution of $G$ independent of absolute time, and wrote $p_g=p_{t,g}$.} On the other hand, summing over $\tau$ gives a relation between $B_t$ and $B_{t,g}^0 = \mathbb{E}_{P_{t,g}}(N^{0,\tau}_{t,g})$:
$$
\begin{array}{rcl}
    B_t(\tau)
     &=& 
     \sum_g p_{t,g}\sum_{\tau'\leq\tau}\mathbb{E}_{P_{t,g}}(n_{t,g}^\tau)\\
      &=& 
     \sum_g p_{t,g}\int_{(0,\tau]}(1-\xi_t F_{t,g}^\text{T}(\tau'))\,\mathrm{d}B^0_{t,g}(\tau')\,.
\end{array}
$$

Then, we can sum up the content of this subsection as follows:

\begin{proposition}
Take $t\in[0,+\infty)$. Assuming the \emph{suppression hypothesis} \eqref{Eq: suppr hyp by sev} and the conditional independence of $\tau^\text{T}$ and $n_t^{0,\tau}$ with respect to $G$, we have
\begin{equation}\label{Eq: suppression of infectiousness}
\begin{array}{rcl}
    \mathbb{E}_{P_t}(n_t^\tau) &=& 
    \sum_g p_{t,g}\,\mathbb{E}_{P_{t,g}}(n_{t,g}^\tau)\\
    &=&
    \sum_g p_{t,g}\,\mathbb{E}_{P_{t,g}}(n_{t,g}^{0,\tau})(1-\xi_t F_{t,g}^\text{T}(\tau))\,,
\end{array}
\end{equation}
where $p_{t,g} := P_t(G_t=g)$. Moreover, 
$$
R_t = \sum_{g\in\mathbb{R}}p_{t,g}R_{t,g} = 
    \sum_g p_{t,g}\int_{\mathbb{R}^+}(1-\xi_t F_{t,g}^\text{T}(\tau))\,\mathrm{d}B^0_{t,g}(\tau)\,.
$$
\end{proposition}

Note that taking the continuum limit of Eq.\ \eqref{Eq: suppression of infectiousness} we retrieve Eq.\ \eqref{Eq: beta_t,g corr}:
$$
\beta_t(\tau) = 
    \sum_g p_{t,g}\beta_{t,g}(\tau) = 
    \sum_g p_{t,g}\,(1-\xi_t F_{t,g}^\text{T}(\tau))\,\beta_{t,g}^0(\tau)\,.
$$

We conclude this subsection by noting, for future convenience, that we can rewrite the suppression hypothesis without referring to $n^{0,\tau}$:
\begin{equation}\label{Eq: suppr hyp 3}
    \mathbb{E}_{P_{t,g}}(n_{t,g}^\tau|\tau_{t,g}^\text{T}=\rho) = 
\frac{1-\xi_t\,\delta_{\tau\geq\rho}}{1-\xi_t F_{t,g}^\text{T}(\tau)}\,\mathbb{E}_{P_{t,g}}(n_{t,g}^\tau)\,.
\end{equation}

%%%

\subsection{Random variables technology}

Given a random variable $X:\Omega\to\mathbb{R}$, it is natural to consider the composition
$$\hat{X} := X\circ\sigma\,.$$
The main use case of this is when $X$ represents the time at which some event related to an individual happens. For example, the infectious age $\tau^\text{T}$ at which they get tested. In this case, $P_t(\hat{\tau}^\text{T}_t=\tau)$ is the probability that, given an individual infected at $t$, their infector is tested at \emph{the infector's} infectious age $\tau$.

In fact, we are often more interested in a slightly different distribution, namely that of the random variable $\check{X}$ defined by
$$
\check{X}(\omega) := 
X(\sigma(\omega)) - \tau^\sigma(\omega)\,.
$$

When $X=\tau^\text{T}$, then $P_t(\check{\tau}^\text{T}_t=\tau)$ is now the probability that, given an individual infected at $t$, their source is tested at \emph{the individual's} infectious age $\tau$.

The next Proposition relates the distributions of $\hat{X}_t$, $\check{X}_t$, and $X_t$:

\begin{proposition}\label{Prop: comp rnd var source}
Take $X:\Omega\to\mathbb{R}$. For all $x\in\mathbb{R}$, we have
$$
\begin{array}{c}
P_t(\hat{X}_t=x) = 
\sum_{\tau\in\mathbb{R}^+}
 P_{t-\tau}(X_{t-\tau}=x)\,
 \mathbb{E}_{P_{t-\tau}}(n_{t-\tau}^\tau | X_{t-\tau} = x)
\frac{\#\Omega_{t-\tau}}{\#\Omega_t}\,,\\
P_t(\check{X}_t=x) = \sum_{\tau\in\mathbb{R}^+}
P_{t-\tau}(X_{t-\tau}=x+\tau)\,
\mathbb{E}_{P_{t-\tau}}(n^\tau_{t-\tau} | X_{t-\tau} = x + \tau)
\frac{\#\Omega_{t-\tau}}{\#\Omega_t}\,.
\end{array}
$$
\end{proposition}

\begin{proof}
We prove both formulas at once by defining 
$X^\alpha:\Omega\to\mathbb{R}$ as
$$
X^\alpha(\omega) := X(\sigma(\omega)) - \alpha\,\tau^\sigma(\omega)
$$
for $\alpha\in\{0,1\}$, so that $X^0=\hat{X}$ and $X^1=\check{X}$.
\begin{multline*}
\#\Omega_t\,P_t(X^\alpha_t=x) = \\
\begin{array}{rcl}
     &=& \#\{\omega\in\Omega_t\ |\ X(\sigma(\omega)) = x + \alpha\,\tau^\sigma(\omega)\}\\
     &=& \#\{(\omega',\omega)\in\tilde\Omega\ |\ \omega\in\Omega_t\,,\ X(\omega') = x + \alpha\,\tau^C(\omega',\omega)\} \\
     &=& \sum_{\tau\in\mathbb{R}^+}\#\{(\omega',\omega)\in\tilde\Omega\ |\ \omega'\in\Omega_{t-\tau}\,,\ X(\omega') = x + \alpha\,\tau\,,\ \omega\in\Omega_t\} \\
    %  &=& \sum_{\tau\in\mathbb{R}^+}\sum_{\omega'\in\Omega_{t-\tau}|X(\omega') = x + \alpha\,\tau}\#\{\omega\in\Omega_t\ |\ \sigma(\omega)=\omega'\} \\
     &=& \sum_{\tau\in\mathbb{R}^+}\sum_{\omega'\in\Omega_{t-\tau}|X(\omega') = x + \alpha\,\tau}
    n^\tau_{t-\tau}(\omega') \\
    &=& \sum_{\tau\in\mathbb{R}^+}
    \mathbb{E}_{P_{t-\tau}}(n^\tau_{t-\tau} | X_{t-\tau} = x + \alpha\,\tau)
    \#\{\omega'\in\Omega_{t-\tau}\ |\ X(\omega') = x + \alpha\,\tau\} \\
&=& \sum_{\tau\in\mathbb{R}^+}\mathbb{E}_{P_{t-\tau}}(n^\tau_{t-\tau} | X_{t-\tau} = x + \alpha\,\tau)P_{t-\tau}(X_{t-\tau}=x+\alpha\,\tau)\#\Omega_{t-\tau}\,.
\end{array}
\end{multline*}

\end{proof}

Notice that we can also break down the right hand side of the formulae of Prop.\ \ref{Prop: comp rnd var source} by the values of $G$. In particular, the second equation can be rewritten as
\begin{equation}\label{Eq: Xcheck by g}
    P_t(\check{X}_t=x) = 
    \sum_{\tau\in\mathbb{R}^+}
    \sum_{g\in\mathbb{R}}
    P_{t-\tau,g}(X_{t-\tau,g}=x+\alpha\,\tau)\,
    \mathbb{E}_{P_{t-\tau,g}}(n^\tau_{t-\tau,g} | X_{t-\tau,g} = x + \alpha\,\tau)
    \frac{\#\Omega_{t-\tau,g}}{\#\Omega_t}\,.
\end{equation}

\subsection{Remarks on generation time and numbers of infected individuals}\label{A: generation time}

In this subsection, we study the distribution of the generation time. First, we do this considering it as a function $\tau^\text{C}:\tilde\Omega\to\mathbb{R}^+$ of infector-infectee pairs, and in particular taking its restriction $\tilde\tau_t^\text{C}$ to the set
$$
\tilde\Omega_t:=\{(\omega',\omega)\in\tilde\Omega\ |\ \omega'\in\Omega_t\}\,.
$$
As usual, on $\tilde\Omega_t$ we put the uniform probability measure, denoted by $\tilde{P}_t$. In this way, we find the intuitive fact that the distribution of $\tau^\text{C}$ restricted to $\tilde\Omega_t$ is just the normalization of the infectiousness:

\begin{proposition}
The CDF of the random variable $\tilde\tau_t^\text{C}:\tilde\Omega_t\to\mathbb{R}^+$ is given by
$$
\tilde{F}^\text{C}_t(\tau) = \frac{B_t(\tau)}{R_t}\,.
$$
\end{proposition}

\begin{proof}
$$\begin{array}{rcl}
    \tilde{F}^\text{C}_t(\tau)
    &=& \tilde{P}_t(\tilde\tau_t^\text{C}\leq\tau) \\
    &=&\frac{1}{\#\tilde\Omega_t}\#\{(\omega',\omega)\in\tilde\Omega_t\ |\ \tau^\text{C}(\omega',\omega)\leq\tau\} \\
    &=&\frac{1}{\#\tilde\Omega_t}
    \sum_{\omega'\in\Omega_t}N^\tau_t(\omega') \\
    &=&\frac{\#\Omega_t}{\#\tilde\Omega_t}\,\mathbb{E}_{P_t}(N^\tau_t)\\
    &=&\frac{\#\Omega_t}{\#\tilde\Omega_t}\,B_t(\tau)\,.
\end{array}$$
The limit $\tau\to+\infty$ gives
$$
\#\tilde\Omega_t = \#\Omega_t\,\mathbb{E}_{P_t}(N^\infty_t) = 
\#\Omega_t\,R_t\,,
$$
and the claim immediately follows.
\end{proof}

Next, we focus on the generation time as a function $\tau^\sigma_t:\Omega_t\to\mathbb{R}^+$ of the infectee. Its probability distribution is given by the formula
$$
P_t(\tau^\sigma_t=\tau) = \frac{\#\Omega_{t-\tau}}{\#\Omega_t}\mathbb{E}_{P_{t-\tau}}(n^\tau_{t-\tau})\,,
$$
which follows from the definitions. Summing the left-hand side over all $\tau>0$ gives 1, from which we find how to compute $\#\Omega_t$ in terms of quantities relative to previous times:
$$
\#\Omega_t = \sum_{\tau\in\mathbb{R}^+}\#\Omega_{t-\tau}\mathbb{E}_{P_{t-\tau}}(n^\tau_{t-\tau})\,.
$$
The last two formulae are easily proven, and the first is an immediate consequence of the next Proposition, which describes the joint probability distribution of $\tau^\sigma_t$ and $\hat{G}_t=(G\circ\sigma)_t$:

\begin{proposition}
$$
P_t(\tau^\sigma_t=\tau, \hat{G}_t=g) =\mathbb{E}_{P_{t-\tau,g}}(n_{t-\tau,g}^\tau)
    \frac{\#\Omega_{t-\tau,g}}{\#\Omega_t}\,.
$$
\end{proposition}

This is the probability that the infector of someone infected at $t$ was infected at $t-\tau$ and had severity $g$.

\begin{proof}
$$\begin{array}{rcl}
\#\Omega_t\,P_t(\tau^\sigma_t=\tau, \hat{G}_t=g)    
    &=& 
    \#\{\omega\in\Omega_t\ |\ \tau^\sigma(\omega)=\tau\,,\ \hat{G}(\omega)=g\}\\
    &=& 
    \#\{(\omega',\omega)\in\tilde\Omega\ |\ \omega\in\Omega_t,\ \omega'\in\Omega_{t-\tau},\ G(\omega')=g)\}\\
    &=& 
    \sum_{\omega'\in\Omega_{t-\tau, g}}n^\tau(\omega')\\
    &=& 
    \#\Omega_{t-\tau, g}\,\mathbb{E}_{P_{t-\tau,g}}(n^\tau_{t-\tau,g})\,.
\end{array}$$
\end{proof}

This formula for the joint probability measure will be used in the next subsection, where we will often use it in integrals. Given that in this paper we always consider $G$ to have a given discrete range, while $\tau^\sigma$ becomes continuous in the continuum limit, we will preferably write these integrals with respect to the improper CDFs
\begin{equation}\label{Eq: joint CDF tausigma}
F^{\sigma,g}_t(\tau) := 
P_t(\tau^\sigma_t\leq\tau, \hat{G}_t=g) = 
\sum_{\tau'\leq\tau}\,\mathbb{E}_{P_{t-\tau',g}}(n_{t-\tau',g}^{\tau'})
\frac{\#\Omega_{t-\tau',g}}{\#\Omega_t}\,.
\end{equation}

% \begin{array}{rcl}
% F^{\sigma,g}_t(\tau) &:=& P_t(\tau^\sigma_t\leq\tau, \hat{G}_t=g)\\
% &=& 
% \sum_{\tau'\leq\tau}\,\mathbb{E}_{P_{t-\tau',g}}(n_{t-\tau',g}^{\tau'})
%     \frac{\#\Omega_{t-\tau',g}}{\#\Omega_t}\\
% &=& 
% \frac{1}{\#\Omega_t}\sum_{\tau'\leq\tau}\,
% \sum_{\omega\in\Omega_{t-\tau',g}} n^{\tau'}(\omega)\,.
% \end{array}
% $$

%%%

\subsection{Time evolution}\label{A: time evolution}

As we saw in \S\ref{tauSt, tauAt, and tauTt}, the key step to determining the time evolution of the system is writing the distribution of the notification time $\tau_t^\text{A,c}$ in terms of that of the testing time $\tau_{t'}^\text{T}$, for $t'<t$. Our assumption is that any infected individual $\omega\in\Omega_t$ is notified precisely at the testing time of their infector $\sigma(\omega)$ with a certain probability $s^\text{c}_t$, provided that such testing time follows the infection time of $\omega$ (otherwise, $\omega$ is never notified). Referring these instants to the infectee's infectious age, we get that $\tau^\text{A,c}(\omega)$ is equal to
$$
\check{\tau}^\text{T}(\omega) = \tau^\text{T}(\sigma(\omega)) - \tau^\sigma(\omega)
$$
with probability $s^\text{c}_t$ in case it is a positive number, and $\tau^\text{A,c}(\omega)=+\infty$ in the remaining cases. This can be written synthetically as
% $$
% P_t(\tau_t^\text{A,c}=\rho) =
% \left\{\begin{matrix}
% 0 & \text{for }\rho\leq 0\\ 
% s_t^\text{c}\,P_t(\check\tau_t^\text{T}=\rho) & \text{for }0<\rho<+\infty\\ 
% P_t(\check\tau_t^\text{T}\leq 0) + (1-s_t^\text{c})P_t(0<\check\tau_t^\text{T}<+\infty) + P_t(\check\tau_t^\text{T}=+\infty) &\text{for }\rho=+\infty
% \end{matrix}\right.
% $$
% In terms of improper CDFs,
\begin{equation}\label{Eq: notif hyp}
    F_t^\text{A,c}(\tau) = 
s^\text{c}_t\,P_t(\check{\tau}_t^\text{T}\in\mathbb{R}^+) = 
s^\text{c}_t\,(\check{F}_t^\text{T}(\tau) - \check{F}_t^\text{T}(0))\,,
\end{equation}
where $F_t^\text{A,c}$ and $\check{F}_t^\text{T}$ are the improper CDFs of $\tau_t^\text{A,c}$ and $\check{\tau}_t^\text{T}$, respectively.

Thus, our goal reduces to computing the distribution of $\check{\tau}^\text{T}$, the (possibly negative) time elapsed from the infectee's contagion to the infector's testing. Applying Eq.\ \eqref{Eq: Xcheck by g} to $X=\tau^\text{T}$ and using the suppression formula \eqref{Eq: suppr hyp 3} we get, for any $\rho\in(-\infty,+\infty]$,
\begin{multline*}
P_t(\check{\tau}_t^\text{T}=\rho) =\\
\begin{array}{rcl}
    &=& 
    \sum_{\tau\in\mathbb{R}^+}\sum_{g\in\mathbb{R}}
    P_{t-\tau,g}(\tau^\text{T}_{t-\tau,g}=\rho+\tau)
    \mathbb{E}_{P_{t-\tau,g}}(n^\tau_{t-\tau,g} | \tau^\text{T} = \rho + \tau)
    \frac{\#\Omega_{t-\tau,g}}{\#\Omega_t}\\
    &=& 
    \sum_{g\in\mathbb{R}}\sum_{\tau\in\mathbb{R}^+}
    P_{t-\tau,g}(\tau^\text{T}_{t-\tau,g}=\rho+\tau)
    \frac{1-\xi_{t-\tau}\,\delta_{\tau\geq\rho+\tau}}
    {1-\xi_{t-\tau} F_{t-\tau,g}^\text{T}(\tau)}\,
    \mathbb{E}_{P_{t-\tau,g}}(n_{t-\tau,g}^\tau)
    \frac{\#\Omega_{t-\tau,g}}{\#\Omega_t}\\
    % &=& 
    % \int_{\mathbb{R}^+\times\mathbb{R}}
    % P_{t-\tau,g}(\tau^\text{T}_{t-\tau,g}=\rho+\tau)
    % \frac{1-\xi_{t-\tau}\,\delta_{\rho\leq0}}
    % {1-\xi_{t-\tau} F_{t-\tau,g}^\text{T}(\tau)}\,
    % \mathrm{d}P_{(\tau^\sigma_t,\hat{G}_t)}(\tau, g)\\
    &=& 
    \sum_{g\in\mathbb{R}}\int_{\mathbb{R}^+}
    P_{t-\tau,g}(\tau^\text{T}_{t-\tau,g}=\rho+\tau)
    \frac{1-\xi_{t-\tau}\,\delta_{\rho\leq0}}
    {1-\xi_{t-\tau} F_{t-\tau,g}^\text{T}(\tau)}\,
    \mathrm{d}F^{\sigma,g}_t(\tau)\,.
    %
    % &=& 
    % \sum_{\tau\in\mathbb{R}^+}
    % (1-\xi\,\delta_{\tau\geq\rho+\tau})\,\sum_g p_{t-\tau,g}P_{t-\tau,g}(\tau_{t-\tau,g}^\text{T}=\rho+\tau)\,
    % \mathbb{E}_{P_{t-\tau,g}}(n_{t-\tau,g}^{0,\tau})\frac{\#\Omega_{t-\tau}}{\#\Omega_t}\\
    % &=& 
    % (1-\xi\,\delta_{\rho\leq 0})\sum_g\sum_{\tau\in\mathbb{R}^+}
    % p_{t-\tau,g}P_{t-\tau,g}(\tau_{t-\tau,g}^\text{T}=\rho+\tau)\,\mathbb{E}_{P_{t-\tau,g}}(n_{t-\tau,g}^{0,\tau})\frac{\#\Omega_{t-\tau}}{\#\Omega_t}\\
    % &=& 
    % (1-\xi\,\delta_{\rho\leq 0})\sum_g\sum_{\tau\in\mathbb{R}^+}
    % p_{t-\tau,g}P_{t-\tau,g}(\tau_{t-\tau,g}^\text{T}=\rho+\tau)\,\mathbb{E}_{P_{t-\tau,g}}(n_{t-\tau,g}^{0,\tau})\frac{\#\Omega_{t-\tau}}{\#\Omega_t}\,.\\
\end{array}
\end{multline*}
Notice that in the last line we used the improper CDFs $F^{\sigma,g}_t$ introduced in Eq.\ \eqref{Eq: joint CDF tausigma}.

Let us try to interpret this formula. The probability distribution of $\check{\tau}^\text{T}$ is obtained by averaging the distributions of $\tau^\text{T}_{t',g}$ for all $t'=t-\tau<t$ and all $g$, each shifted by $\tau$ to the left to account for the switch from the infector's to the infectee's infectious age. This averaging is done by integrating over all $t'$ and $g$ with respect to the joint distribution of the generation time $\tau^\sigma$ and the infector's severity $\hat{G}$. But a correction factor (the fraction) appears in the integral, as the fact that the infector infects at relative time $\tau$ and has severity $g$ conditions the distribution of $\tau^\text{T}_{t-\tau,g}$, by shifting it toward values greater than $\tau$. Indeed, the correction factor is greater than 1 for $\rho>0$ and less than 1 otherwise. This means that, compared to a hypothetical case in which the testing time and the infectiousness are independent (which happens when $\xi$ is constantly zero), the probability $P_t(\check{\tau}_t^\text{T}=\rho)$ is higher after the contagion time (i.e.,\ when $\rho>0$) and lower before.

It follows now that the improper CDF of $\check{\tau}_t^\text{T}$ reads
$$
\check{F}^\text{T}_t(\rho) =
\left\{\begin{matrix}
\sum_{g\in\mathbb{R}}\int_{\mathbb{R}^+}
    F^\text{T}_{t-\tau,g}(\rho + \tau)
    \frac{1-\xi_{t-\tau}}
    {1-\xi_{t-\tau} F_{t-\tau,g}^\text{T}(\tau)}\,
    \mathrm{d}F^{\sigma,g}_t(\tau) & \text{ for }\rho\leq 0\,,\\ 
\check{F}^\text{T}_t(0) + \sum_{g\in\mathbb{R}}\int_{\mathbb{R}^+}
    \frac{F^\text{T}_{t-\tau,g}(\rho + \tau) - F^\text{T}_{t-\tau,g}(\tau)}
    {1-\xi_{t-\tau} F_{t-\tau,g}^\text{T}(\tau)}\,
    \mathrm{d}F^{\sigma,g}_t(\tau) & \text{ for }0<\rho<+\infty\,.\\ 
\end{matrix}\right.
$$
We only have to replace this equation in \eqref{Eq: notif hyp} to get the time evolution formula:

\begin{proposition}
Assuming the \emph{suppression hypothesis} \eqref{Eq: suppr hyp by sev}, the conditional independence of $\tau^\text{T}$ and $n_t^{0,\tau}$ with respect to $G$, and the notification hypothesis \eqref{Eq: notif hyp}, we have
$$
F_t^\text{A,c}(\rho) = 
s^\text{c}_t\,\sum_{g\in\mathbb{R}}\int_{\mathbb{R}^+}
    \frac{F^\text{T}_{t-\tau,g}(\rho + \tau) - F^\text{T}_{t-\tau,g}(\tau)}
    {1-\xi_{t-\tau} F_{t-\tau,g}^\text{T}(\tau)}\,
    \mathrm{d}F^{\sigma,g}_t(\tau)
$$
for $\rho>0$ and $F_t^\text{A,c}(\rho) = 0$ otherwise.
\end{proposition}

\subsection{Modifications in the case of use of a contact tracing app}\label{A: scenario with app}

The inhomogeneity in the population due to the use of a contact tracing app by a part of it can be partly addressed in an analogous way to the inhomogeneity due to different degrees of severity of the illness. Namely, we introduce a new random variable
$$
A:\Omega\to\{\text{app},\text{no app}\}
$$
whose value determines whether or not an individual $\omega\in\Omega$ has the app. We assume that whether or not an individual has the app is independent of both their severity and their infectiousness \emph{in the absence} of measures. In other words, $A$ is independent of $G$ and $n^{0,\tau}$, for all $\tau$. On the other hand, the infectiousness (in presence of measures) and the testing time of an individual will be different depending on whether or not they use the app. 

$A$ further partitions $\Omega$: we write
$$
\Omega_{t,g,a}:=\{\omega\in\Omega_t\ |\ G(\omega)=g\,,\ A(\omega)=a\}\,,\quad
n_{t,g,a}^\tau := n^\tau{\restriction_{\Omega_{t,g,a}}}\,,
$$
and so on. The content of \S\ref{A: suppression formula} fully applies to this scenario, but we want now to have formulae conditioned on $A$. As $A$ and $n^{0,\tau}$ are independent, the previous formulae simply become
% $$
% \mathbb{E}_{P_{t,g,a}}(n_{t,g,a}^\tau|\tau_{t,g,a}^\text{T}=\rho) =
% (1-\xi_t\,\delta_{\tau\geq\rho})\,\mathbb{E}_{P_{t,g}}(n_{t,g}^{0,\tau})
% $$
% and then
$$
\mathbb{E}_{P_{t,g,a}}(n_{t,g,a}^\tau) =
    \mathbb{E}_{P_{t,g}}(n_{t,g}^{0,\tau})(1-\xi_t F_{t,g,a}^\text{T}(\tau))
$$
and
\begin{equation}\label{Eq: suppr formula with app}
\begin{array}{rcl}
    \mathbb{E}_{P_{t,g,a}}(n_{t,g,a}^\tau|\tau_{t,g,a}^\text{T}=\rho) &=& 
(1-\xi_t\,\delta_{\tau\geq\rho})\,\mathbb{E}_{P_{t,g}}(n_{t,g}^{0,\tau}) \\
    &=&
    \frac{1-\xi_t\,\delta_{\tau\geq\rho}}{1-\xi_t F_{t,g,a}^\text{T}(\tau)} \mathbb{E}_{P_{t,g,a}}(n_{t,g,a}^\tau)\,.
\end{array}    
\end{equation}
The suppression formula for $R_t$ can be broken down to
$$
\begin{array}{rcl}
    R_t
    &=& 
    \sum_{g\in\mathbb{R}}
    \sum_a p_{t,g}\epsilon_{t,a} R_{t,g,a}\\
    &=& 
    \sum_g\sum_a p_{t,g}\epsilon_{t,a} \int_{\mathbb{R}^+}
    (1-\xi_t\, F_{t,g,a}^\text{T}(\tau))\,
    \mathrm{d}B^0_{t,g}(\tau)\,,
\end{array}
$$
where $\epsilon_{t,a}:=P_t(A_t=a)$.

The time evolution equation has to be treated differently, as the receipt of the notification depends on whether both the infector and the infectee use the app.

Let $\check{F}_t^{\text{T},a}(\tau)$ denote the probability that, given an individual $\omega\in\Omega_t$, the infector $\sigma(\omega)$ is tested at a time $\leq t+\tau$ and we have $\hat{A}(\omega)=A(\sigma(\omega))=a$.

According to our assumptions, given an infection occurred at $t$, the probability that the infector notifies the infectee when they test positive (provided that this happens after the infection) is $s^\text{c,app}_t$ in the case that both individuals have the app, and is $s^\text{c,no app}_t$ otherwise. Therefore, the contact tracing hypothesis \eqref{Eq: notif hyp} is now replaced by the following expressions for the CDFs of the time of the notification received by an individual with or without the app, respectively:
\begin{equation}\label{Eq: contact tracing formula with app}
\begin{array}{rcl}
    F^\text{A,c}_{t,\text{app}}(\rho) &=& s^\text{c,app}_t
    (\check{F}_t^\text{T,app}(\rho) - \check{F}_t^\text{T,app}(0))
    + s^\text{c,no app}_t
    (\check{F}_t^\text{T,no app}(\rho) - \check{F}_t^\text{T,no app}(0))\,,  \\
    F^\text{A,c}_{t,\text{no app}}(\rho) &=& s^\text{c,no app}_t
    (\check{F}_t^\text{T,app}(\rho) -
    \check{F}_t^\text{T,app}(0)
    + \check{F}_t^\text{T,no app}(\rho) -
    \check{F}_t^\text{T,no app}(0)
    ) \\
    &=& s^\text{c,no app}_t
    (\check{F}_t^\text{T}(\rho) -
    \check{F}_t^\text{T}(0)
    )\,.
\end{array}
\end{equation}

Now, the improper CDFs $\check{F}_t^{\text{T},a}$ can be computed just as before, simply treating the conditioning on $\hat{A}$ as we treated the conditioning on $\hat{G}$:
\begin{equation}\label{Eq: joint prob checktauT, G, A}
\begin{array}{l}
    P_t(\check{\tau}_t^\text{T}=\rho, \hat{G}_t=g, \hat{A}_t=a)\\
    \quad  = 
\sum_{\tau>0}
P_{t-\tau,g,a}(\tau^\text{T}_{t-\tau,g,a}=\rho+\tau)\,
\mathbb{E}_{P_{t-\tau,g,a}}(n^\tau_{t-\tau,g,a} | \tau^\text{T}=\rho+\tau)\frac{\#\Omega_{t-\tau,g,a}}{\#\Omega_t}\\
    \quad  = 
    \int_{\mathbb{R}^+} 
    P_{t-\tau,g,a}(\tau^\text{T}_{t-\tau,g,a}=\rho+\tau)\,
\mathbb{E}_{P_{t-\tau,g,a}}(n^\tau_{t-\tau,g,a} | \tau^\text{T}=\rho+\tau)\,
\mathrm{d}F^{\sigma,g,a}_t(\tau)
\,,
\end{array}
\end{equation}
where we defined $F^{\sigma,g,a}_t$ as follows, proceeding like in \S\ref{A: generation time} to compute the joint distribution of $\tau^\sigma_t$, $\hat{G}_t$, and $\hat{A}_t$:
\begin{equation}\label{Eq: joint CDF tausigma G, A}
F^{\sigma,g,a}_t(\tau) := 
P_t(\tau^\sigma_t\leq\tau, \hat{G}_t=g, \hat{A}_t=a) = 
\sum_{\tau'\leq\tau}\,
\mathbb{E}_{P_{t-\tau',g,a}}(n_{t-\tau',g,a}^{\tau'})
\frac{\#\Omega_{t-\tau',g,a}}{\#\Omega_t}\,.
\end{equation}
It is worth noting that comparing this equation with \eqref{Eq: joint CDF tausigma} we get
\begin{equation}\label{Eq: Fsigma a in terms of Fsigma}
\frac{\mathrm{d}F^{\sigma,g,a}_t(\tau)}
{1-\xi_{t-\tau} F_{t-\tau,g,a}^\text{T}}
=
\epsilon_{t,a}\frac{\mathrm{d}F^{\sigma,g}_t(\tau)}
{1-\xi_{t-\tau} F_{t-\tau,g}^\text{T}}\,.
\end{equation}

Replacing the suppression formula \eqref{Eq: suppr formula with app} in \eqref{Eq: joint prob checktauT, G, A} and summing over $\rho$, we end up with
$$
\check{F}_t^\text{T,a}(\rho) - \check{F}_t^\text{T,a}(0) = 
\sum_g\int_{\mathbb{R}^+}
    \frac{F^\text{T}_{t-\tau,g,a}(\rho + \tau) - F^\text{T}_{t-\tau,g,a}(\tau)}
    {1-\xi_{t-\tau} F_{t-\tau,g,a}^\text{T}(\tau)}\,
    \mathrm{d}F^{\sigma,g,a}_t(\tau)
$$
for $\rho>0$. Plugging this into \eqref{Eq: contact tracing formula with app} gives us $F^{\text{A,c}}_{t,a}$ in terms of $F^{\sigma,g,a}_t$ and $F^\text{T}_{t',g,a}$ for $t'<t$, that is the time evolution equation for the scenario with app usage. For $a=\text{app}$, this is Eq.\ \eqref{Eq: time evolution equation with app}, while for $a=\text{no app}$ it simplifies to Eq.\ \eqref{Eq: time evolution equation without app}, since when the infectee doesn't have the app it is irrelevant whether or not the infector has the app. This is evident from the last line of Eq.\ \eqref{Eq: contact tracing formula with app}, which in fact could also have been used, together with the expression for $\check{F}_t^\text{T}$ derived in \S\ref{A: time evolution}, to get Eq.\ \eqref{Eq: time evolution equation without app}. Using Eq.\ \eqref{Eq: Fsigma a in terms of Fsigma}, it can be checked immediately that the two approaches give the same result.

\begin{acknowledgements}
This work was supported by the authors’ employer, Bending Spoons S.p.A, which was involved in the development of the contact tracing app adopted by the Italian Government. We thank our colleagues Christy Keenan and Luca Ferrari for carefully proofing the manuscript and providing significant stylistic improvements. We are also grateful to Giorgio Guzzetta and the anonymous referees for several useful comments and suggestions about this study and its publication.
\end{acknowledgements}

\begin{contributions}
M.M.\ and M.R.\ conceived the study and its main ideas; A.M.\ developed the mathematical model; A.M.\ and M.M.\ developed the Python repository with the calculations and wrote the manuscript.
\end{contributions}

% BIBLIOGRAPHY

\bibliographystyle{alpha}

\bibliography{Bib-VarPap}

\newcommand{\etalchar}[1]{$^{#1}$}
\begin{thebibliography}{FWK{\etalchar{+}}20}

\bibitem[BMS{\etalchar{+}}21]{bendavid2020covid}
Eran Bendavid, Bianca Mulaney, Neeraj Sood, Soleil Shah, Rebecca
  Bromley-Dulfano, Cara Lai, Zoe Weissberg, Rodrigo Saavedra-Walker, Jim
  Tedrow, Andrew Bogan, Thomas Kupiec, Daniel Eichner, Ribhav Gupta, John P~A
  Ioannidis, and Jay Bhattacharya.
\newblock {{COVID-19} antibody seroprevalence in {Santa Clara County,
  California}}.
\newblock {\em International Journal of Epidemiology}, 02 2021.

\bibitem[FRAF04]{fraser2004factors}
Christophe Fraser, Steven Riley, Roy~M Anderson, and Neil~M Ferguson.
\newblock Factors that make an infectious disease outbreak controllable.
\newblock {\em Proceedings of the National Academy of Sciences},
  101(16):6146--6151, 2004.

\bibitem[FWK{\etalchar{+}}20]{FerrEtAl2020Quant}
Luca Ferretti, Chris Wymant, Michelle Kendall, Lele Zhao, Anel Nurtay, Lucie
  Abeler-D{\"o}rner, Michael Parker, David Bonsall, and Christophe Fraser.
\newblock Quantifying {SARS-CoV-2} transmission suggests epidemic control with
  digital contact tracing.
\newblock {\em Science}, 2020.

\bibitem[LGW{\etalchar{+}}20]{li2020early}
Qun Li, Xuhua Guan, Peng Wu, Xiaoye Wang, Lei Zhou, Yeqing Tong, Ruiqi Ren,
  Kathy~S.M. Leung, Eric~H.Y. Lau, Jessica~Y. Wong, et~al.
\newblock Early transmission dynamics in {Wuhan, China}, of novel
  coronavirus--infected pneumonia.
\newblock {\em New England Journal of Medicine}, 2020.

\bibitem[MKD00]{muller2000contact}
Johannes M{\"u}ller, Mirjam Kretzschmar, and Klaus Dietz.
\newblock Contact tracing in stochastic and deterministic epidemic models.
\newblock {\em Mathematical Biosciences}, 164(1):39--64, 2000.

\bibitem[MKZC20]{mizumoto2020estimating}
Kenji Mizumoto, Katsushi Kagaya, Alexander Zarebski, and Gerardo Chowell.
\newblock Estimating the asymptomatic proportion of coronavirus disease 2019
  ({COVID-19}) cases on board the {Diamond Princess} cruise ship, {Yokohama,
  Japan}, 2020.
\newblock {\em Eurosurveillance}, 25(10):2000180, 2020.

\bibitem[MM21]{mm2021epidemics}
Andrea Maiorana and Marco Meneghelli.
\newblock Epidemics suppression model.
\newblock \url{https://github.com/BendingSpoons/epidemic-suppression-model},
  2021.

\bibitem[PDG20]{OxfordRepo}
Oxford~University Pathogen Dynamics~Group, Big Data~Institute.
\newblock \url{https://github.com/BDI-pathogens/OpenABM-Covid19}, 2020.

\end{thebibliography}

\end{document}